\documentclass{article}

\usepackage[T1]{fontenc}
\usepackage{graphicx}
\usepackage{adjustbox}
\usepackage{tabularx}
\usepackage{amsmath,bm}
\usepackage{multirow}
\usepackage{comment}
\usepackage{amsthm}
\usepackage{blindtext}
\usepackage{float}
\usepackage{enumitem}
\usepackage{multirow}
\usepackage{bm}
\usepackage{algorithm, algorithmic}
\usepackage{textcomp}
\usepackage{appendix}
\usepackage{amsfonts}
\usepackage{bbding}
\usepackage{amsthm,mathtools}
\usepackage{hyperref}
\usepackage{xcolor}
\usepackage{fullpage}
\usepackage{color,soul} 
\DeclareMathOperator{\Tr}{Tr}

\newtheorem{Ex}{Example}
\newtheorem{Cor}{Corollary}
\newtheorem{Def}{Definition}
\newtheorem{Thm}{Theorem}
\newtheorem{Lem}{Lemma}
\newtheorem{Rmk}{Remark}
\providecommand{\Keywords}[1]{\textbf{\textit{Keywords:}} #1}

\begin{document}
\title{PAC Privacy: Automatic Privacy Measurement and Control of Data Processing}
%

\author{Hanshen Xiao \\ \href{mailto:hsxiao@mit.edu}{hsxiao@mit.edu} 
   \and Srinivas Devadas \\ \href{mailto:devadas@mit.edu}{devadas@mit.edu}}

%
%

\newcommand{\hanshen}[1]{\textbf{\color{blue}Hanshen:~}{\color{blue}#1}}

\maketitle              

\begin{abstract}
We propose and study a new privacy definition, termed Probably Approximately Correct (PAC) Privacy. PAC Privacy characterizes the information-theoretic hardness to recover sensitive data given arbitrary information disclosure/leakage during/after any processing. Unlike the classic cryptographic definition and Differential Privacy (DP), which consider the {\em adversarial $($input-independent$)$ worst case}, PAC Privacy is a simulatable metric that quantifies {\em the instance-based} impossibility of inference. A fully automatic analysis and proof generation framework is proposed: 
security parameters can be produced with arbitrarily high confidence via Monte-Carlo simulation for any black-box data processing oracle. This appealing automation property enables analysis of complicated data processing, where the worst-case proof in the classic privacy regime could be loose or even intractable. Moreover, we show that the produced PAC Privacy guarantees enjoy simple composition bounds and the automatic analysis framework can be implemented in an online fashion to analyze the composite PAC Privacy loss even under correlated randomness. On the utility side, the magnitude of (necessary) perturbation required in PAC Privacy is {\em not} lower bounded by ${\Theta}(\sqrt{d})$ for a $d$-dimensional release but could be $O(1)$ for many practical data processing tasks, which is in contrast to the input-independent worst-case information-theoretic lower bound. Example applications of PAC Privacy are included with comparisons to existing works.  

\end{abstract}

\section{Introduction}
\label{sec: Intro}
Privacy concerns with information leakage from data processing are receiving increasing attention. Obtaining usable security which provides high utility, low implementation overhead and meaningful interpretation, is challenging. Theoretically, the underlying problem can be described by the following generic model. Assume some sensitive data $X$ is processed by some mechanism $\mathcal{M}$ and the output $\mathcal{M}(X)$ is published. We want to characterize, given $\mathcal{M}(X)$, how much information an adversary can infer about $X$. In the applications of data analysis and machine learning, $\mathcal{M}(X)$ could be any statistics or model learned from $X$ \cite{DP-deeplearning}; in encryption, $\mathcal{M}(X)$ could be ciphertexts of $X$ \cite{goldwasser1982probabilistic}; in side channel attacks, $\mathcal{M}(X)$ could be any physical signal observed by the adversary, such as memory access pattern \cite{oram} or network traffic \cite{chen2010side} when accessing or sending $X$.

To formalize and quantify the information leakage, a large number of security metrics have been proposed varying from the perfect secrecy \cite{shannon1949communication} of Shannon, to computational indistinguishability \cite{goldwasser1984probabilistic}, which forms the foundation of modern cryptography, to Differential Privacy (DP) \cite{dwork2006calibrating} and its variants based on 
various statistical divergence metrics \cite{mironov2017renyi}, and many other definitions based on information-theoretic quantities, such as Fisher-information privacy \cite{fisher2017}. 

Most existing privacy or security definitions (including all above-mentioned examples)  are rooted in indistinguishability likelihoods. They set out to quantify the difference between the likelihoods $\mathbb{P}(\mathcal{M}(X)=o|X=X_0)$ and $\mathbb{P}(\mathcal{M}(X)=o|X=X'_0)$ for different input selections from various perspectives. For instance, perfect secrecy requires that for arbitrary input data $X$, the distribution of $\mathcal{M}(X)$ must be identical. Cryptographic security relaxes 
the statistical indistinguishability to that any polynomial-time adversary cannot distinguish the distributions of $\mathcal{M}(X_0)$ and $\mathcal{M}(X'_0)$ for any input candidates $X_0$ and $X'_0$. Similarly, DP and its variants quantify the difference between $\mathcal{M}(X_0)$ and $\mathcal{M}(X'_0)$ for any two {\em adjacent datasets} \footnote{Here, we say $X_0$ and $X'_0$ are adjacent if they only differ in one datapoint, i.e., their Hamming distance is $1$. } by some divergence function $\mathcal{D} \big(\mathsf{P}_{\mathcal{M}(X_0)} \| \mathsf{P}_{\mathcal{M}(X'_0)} \big)$. For example, classic $\epsilon$-DP \cite{dwork2006calibrating} is based on maximal divergence $\mathcal{D}_{\infty}(\mathsf{P}_{\mathcal{M}(X_0)} \| \mathsf{P}_{\mathcal{M}(X'_0)}) = \sup_{o} \log \big({\mathbb{P}(\mathcal{M}(X_0)=o)}/{\mathbb{P}(\mathcal{M}(X'_0)=o)} \big)$.  

The interested reader may wonder if there is an intuitive privacy interpretation of the limited likelihood difference? In particular, can likelihood indistinguishability tell us given $\mathcal{M}(X)$, how much of the sensitive data $X$ an adversary can recover? Unfortunately, none of above-mentioned definitions can directly answer this question.  This is {\em not} because those metrics are incomplete but because additional assumptions must be made. In general, to quantify reconstruction hardness, we have to specify the adversary's prior knowledge 
and the (conditional) entropy of objective $X$. In the extreme case where the adversary already has the full knowledge of $X$, or where $X$ is simply a public constant, then no matter what kind of mechanism $\mathcal{M}$ is applied to encode $X$, the adversary can perfectly recover it, regardless of $\mathcal{M}(X)$. Thus, {\em the likelihood difference only partially captures the changes to the adversary's belief or the additional information regarding $X$ after observing  $\mathcal{M}(X)$.} 

Compared to the extensive studies along the line of (in)distinguishability, formal characterization of reconstruction hardness is a long-standing, challenging problem. The estimation of the priors regarding sensitive data $X$ and resultant output distribution of $\mathcal{M}(X)$ could already be intractable -- this is especially likely in the high-dimensional scenario. Thus, as mentioned before, most security/privacy definitions only focus on the likelihood divergence and avoid complicated, such as Bayesian, analysis to properly quantify the posterior inference hardness. As a consequence, to make a meaningful guarantee,  most classic definitions treat the sensitive input {\em deterministically} and consider the adversarial worst-case, i.e., the maximal distribution divergence between $\mathcal{M}(X_0)$ and $\mathcal{M}(X'_0)$ for two arbitrary input candidates $X_0$ and $X'_0$.  Such worst-case indistinguishability is a very strong guarantee since it is independent of any particular input distribution, but is also expensive, especially for statistical data processing. At a high level, there are two key obstacles in the application of such indistinguishability measurement:

\begin{enumerate}[label=(\alph*)]
    \item \textbf{Adversarial Worst-case Proof is Hard}: It is, in general, non-trivial to tightly or analytically characterize the worst case and privatize an  (adaptively-iterative) algorithm to ensure satisfied output divergence. For example, in DP, calculation of the sensitivity, the maximal possible change on the output by replacing a single datapoint, is, in general, NP-hard \cite{sensitivity-NPhard}. Even in common applications such as to differentially privately train a neural network, the tight sensitivity bound of Stochastic Gradient Descent (SGD) remains open. The only known generic privatization approach, DP-SGD, restricts the sensitivity using clipped/normalized per-sample gradient, and adds noise step by step to derive an upper bound in a compositional fashion \cite{DP-deeplearning}. However, the tight composition bound calculation is also known to be \#P hard \cite{murtagh2016complexity}. Though many efforts have been devoted to obtaining efficient and tighter composition \cite{mironov2017renyi}, the applications of DP-SGD in modern deep learning, and large-scale private database queries still encounter severe utility compromise. Therefore, an automatic and usable privacy analysis framework for arbitrary 
    data processing is desirable. 

\item \textbf{Data Entropy Matters in Reconstruction}: Reconstruction (or inference) of data is, in general, a harder problem for an adversary than simply distinguishing \cite{fisher2022Guo}. However, worst-case indistinguishability cannot capture the data entropy. We give a simple example. A dataset $X = \{x_1, x_2, \ldots , x_n\}$ is generated as  i.i.d. samples from some Gaussian distribution and we simply publish its empirical average, i.e., $\mathcal{M}(X) = 1/n \cdot (\sum_{i=1}^n x_i)$. Even without any further obfuscation on $\mathcal{M}(X)$ like perturbation, intuitively, we see some potential hardness to reconstruct $X$ from its average $\mathcal{M}(X)$. For an adversary who only knows the data distribution, it is impossible to perfectly reconstruct $X$ deterministically due to ambiguity. But the data leakage of this trivial aggregation cannot be properly analyzed from a distinguishability angle, since one can always determine the true input from two datasets whose means are different. However, in practical processing even for a given data pool, many kinds of randomness can enforce data to be processed to enjoy entropy, for example, subsampling or data augmentation \cite{data_augmentation}, where the participants are randomized. Therefore, a more generic security/privacy definition, which enables {\em instance-based analysis} in terms of both data priors and specific inference objective, is attractive. 
\end{enumerate}
  
In this paper, we tackle the fundamental problem of quantifying the relationship among data entropy, disclosure, and reconstruction hardness. Besides a more meaningful interpretation via inference impossibility, to address the above two mentioned challenges in classic security/privacy regimes, we want a more generic definition, which enables (a) automatic security proof, and (b) instance-based worst-case analysis. To this end, a new privacy definition, termed Probably Approximately Correct (PAC) Privacy, is proposed, which enables a user to keep track of the data leakage of an arbitrary (black-box) processing mechanism with confidence, via simulation. At a high level, $(\delta, \rho, \mathsf{D})$ PAC Privacy characterizes the following reconstruction hardness: 

{\em For some sensitive data $X$ distributed in $\mathsf{D}$ and a mechanism $\mathcal{M}$, given $\mathcal{M}(X)$, there does not exist an adversary $($possibly computationally-unbounded$)$ who can  successfully return $\tilde{X}$ such that $\rho(\tilde{X}, X) = 1$ under measure function $\rho$ with probability at least $(1-\delta)$.} 

{We use the notion {\em “PAC”} as we borrow the idea of well-known PAC learning theory \cite{PAC} and describe the adversary's reconstruction/inference task as a generic learning problem. The difference compared to classic learning theory is that we study impossibility results.} {The $\rho(\cdot,\cdot)$ captures an arbitrary inference task of interest and  $\rho(\tilde{X},X)=1$ if and only if adversary's inference $\tilde{X}$ satisfies the reconstruction criterion. For example, if we set out to prevent the adversary from recovering a single digital bit of data $X$, then we may select $\rho$ such that $\rho(\tilde{X},X)$ returns $1$ only when $\tilde{X}$ and $X$ collide in at least one coordinate. More discussions can be found below Definition {\ref{def:pac}}.} 

Our presented PAC Privacy framework has the following appealing properties and important implications:
\begin{enumerate}
    \item \textbf{Wide Applicability}: The PAC Privacy framework presented can theoretically be used to analyze any data processing.  Via $f$-divergence, we present a simple quantification of the posterior advantage that the adversary gains by observing the disclosure $\mathcal{M}(X)$ for arbitrary inference problems. Thus, PAC Privacy can characterize very generic reconstruction/inference hardness under any selection of the measure function $\rho$, which captures classic attacks such as membership inference and identification. Through instance-based (Gaussian) perturbation, we show how to efficiently control PAC Privacy parameters. In contrast to the adversarial worst-case, the noise magnitude in PAC Privacy does not necessarily scale with dimensionality but can be $O(1)$ in high-dimensional release.
    \item \textbf{Automatic Privacy Analysis}: To our knowledge, PAC Privacy is the first non-trivial automatic security definition. Theoretically, no worst-case proof is needed but one can determine a perturbation mechanism (if necessary) fully automatically in polynomial time to ensure desired security parameters with arbitrarily high confidence. More importantly, this universal framework only requires sufficient Monte-Carlo simulations while the underlying data processing mechanism could be a black-box oracle -- no algorithmic analysis is required. This is very different from classic definitions such as cryptography and DP, where the key challenge is to derive the algorithmic worst-case proof.

    \item \textbf{Connections to the Input-Independent Worst Case and Generalization Error}: We also connect PAC Privacy, DP and generalization error. Given assumptions on priors, we show DP can produce PAC Privacy, and PAC Privacy can then provide a generalization error guarantee, while the reverse direction, in general, does not hold. However, in practical data applications, PAC Privacy could bring significant (empirical) perturbation improvement, and can characterize the optimal instance-based perturbation required. We show concrete examples in private release of local samples, mean estimation and deep learning models. 
    {Though the tools we apply to develop the foundation of PAC Privacy include $f$-divergence and mutual information, we stress that PAC Privacy is {\em not} defined using any existing information-theoretical metric or statistical divergence, but by the impossibility of inference described above. Therefore, PAC Privacy enjoys a very straightforward interpretation. More discussion 
    can be found in Section {\ref{sec: related-work}}.  }
\end{enumerate}

\noindent \textbf{Paper Organization}: At a high level, the paper has two parts, i.e., the definition of PAC Privacy and the control of security parameters via automatic analysis. In Section \ref{sec:pac_security}, we formally present the definitions and interpretations of PAC Privacy. In Section \ref{sec:bound}, we show generic bounds of PAC Privacy parameters using $f$-divergence and mutual information in various setups. With these bounds, we move on to present automatic algorithms to determine proper perturbation (if necessary) to produce desired PAC Privacy. In Section \ref{sec:deter_alg}, we present the methods for deterministic data processing, and in Section \ref{sec:random_alg}, we study the case of randomized mechanisms. In Section \ref{sec:optimal}, we discuss how to approximate the optimal perturbation in general. In Section \ref{sec: composition}, we study the composition of PAC Privacy in various setups. In Section \ref{sec:worst_vs_average}, we proceed to study and compare the adversarial and instance-based worst case. Practical applications of PAC Privacy are included. We summarize related works in Section \ref{sec: related-work}, where we also discuss the relationship between PAC Privacy, Differential Privacy and generalization error.  Finally, we conclude in Section \ref{sec:conclusion}. 
The generalization to local PAC Privacy via secure implementation of automatic privacy analysis in a decentralized setup can be found in the full version.  

\section{PAC Privacy}
\label{sec:pac_security}
We assume each private datapoint $x$ is defined over some domain $\mathcal{X}$, and the dataset $X \in \mathcal{X}^*$. Our goal is to formalize the privacy leakage from a data processing mechanism $\mathcal{M}$, where $\mathcal{M}: \mathcal{X}^* \to \mathcal{Y} \subset \mathbb{R}^d$. Both $\mathcal{X}$ and $\mathcal{Y}$ are measurable spaces. The data distribution $\mathsf{D}$ is defined over $\mathcal{X}^*$. $\mathcal{M}$ can be either deterministic or randomized. Different analyses will be provided later for these two cases. The formal definition of PAC Privacy is as follows. 

\begin{Def}[$(\delta, \rho, \mathsf{D})$ PAC Privacy]
For a data processing mechanism $\mathcal{M}$, given some data distribution $\mathsf{D}$, and a measure function $\rho(\cdot, \cdot)$, we say $\mathcal{M}$ satisfies $( \delta, \rho, \mathsf{D})$-PAC Privacy if the following experiment is impossible:

A user generates data $X$ from distribution $\mathsf{D}$ and sends $\mathcal{M}(X)$ to an adversary. The adversary who knows $\mathsf{D}$ and $\mathcal{M}$ is asked to return an estimation $\tilde{X} \in \mathcal{X}^*$ on $X$ such that with probability at least $(1-\delta)$, $\rho(\tilde{X}, X) =1$. 
\label{def:pac}
\end{Def}

In Definition \ref{def:pac}, we adopt a way to describe the reconstruction hardness as a game between the user and the adversary. { The adversary sets out to learn something from the disclosure $\mathcal{M}(X)$. The measure function (on the adversary's learning performance) $\rho(\cdot, \cdot)$ can be arbitrarily selected according to the adversarial inference task of interest. For example, $\mathcal{X} \subset \mathbb{R}^{d'}$ for some dimension $d'$ and a safe reconstruction error is set to be $\epsilon$ in the $l_2$ norm. Then, we can define $\rho(\tilde{X}, X)=1$ if and only if $\|\tilde{X}-X\|_2 \leq \epsilon$ and Definition \ref{def:pac} suggests it is impossible for an adversary to recover private data $X$ within $\epsilon$ error with probability at least $(1-\delta)$. Another example, $\mathcal{X}^*$ is some finite set and $X \in \mathcal{X}^*$ contains $n$ elements. Similarly, we define $\rho(\tilde{X}, X)=1$ iff $\#(\tilde{X}\cap X) > (n-\epsilon)$. Then, Definition \ref{def:pac} implies that there does not exist an adversary who can identify more than $(n-\epsilon)$ elements/participants of $X$ from $\mathcal{X}^*$ with probability at least $(1-\delta)$. } As a remark, a neat expression of $\rho$ is not important or even necessary in PAC Privacy and we simply use $\rho$ to indicate whether the adversary's reconstruction satisfies a certain criterion.  

It is worthwhile to mention that in Definition \ref{def:pac} and results presented later, we do not put any specific restrictions on $\mathsf{D}$ or the adversary's strategy.  $\mathsf{D}$ could be a generic joint distribution of datapoints generated, though with additional assumptions we may obtain tighter bounds (such as Theorem \ref{thm:i.i.d.pac}).  Before proceeding, we augment the definition of PAC Privacy to incorporate (instance-based) additional {\em posterior advantage}, commonly adopted in classic privacy definitions. That is, for any specific priors given, we set out to quantify how much the disclosure of $\mathcal{M}(X)$ can help the adversary to implement successful inference. To measure such advantage, we first formally introduce the definition of $f$-divergence. 

\begin{Def}[$f$-Divergence] Let $f: (0, +\infty) \to \mathbb{R}$ be a convex function with $f(1)=0$. Let $\mathsf{P}$ and $\mathsf{Q}$ be two distributions on some measurable space, and the $f$-divergence $\mathcal{D}_f$ between $\mathsf{P}$ and $\mathsf{Q}$ is defined as 
$$\mathcal{D}_f(\mathsf{P}\|\mathsf{Q}) \coloneqq \mathbb{E}_{Q}\big[ f(\frac{d\mathsf{P}}{d\mathsf{Q}}) \big].$$
Here, we use $d\mathsf{P}$ $(d\mathsf{Q})$  to represent the probability density $($or mass$)$ function of the continuous $($or discrete$)$ probability distribution $\mathsf{P}$ $(\mathsf{Q})$. 
\label{def: f-divergence}
\end{Def}
In Definition \ref{def: f-divergence}, when we select $f(x) = x\log(x)$, it becomes the Kullback-Leibler (KL) divergence; when we select $f(x)=\frac{1}{2}|x-1|$, it becomes the total variation, where $\mathcal{D}_{TV}(\mathsf{P}\|\mathsf{Q}) = \frac{1}{2}\cdot \int |d\mathsf{P}-d\mathsf{Q}|$. In the following, we define $(1-\delta^{\rho}_o)$ as the optimal {\em a priori} chance that an adversary who only knows the data distribution $\mathsf{D}$ can complete the objective reconstruction task {\em before} observing the release $\mathcal{M}(X)$. 

\begin{Def}[Optimal Prior Success Rate] For given measure $\rho$ and data distribution $\mathcal{D}$, we define $(1-\delta^{\rho}_o)$ as
$$ 1-\delta^{\rho}_o \coloneqq  \sup_{\tilde{X} \in \mathcal{X}^*} \Pr_{X \sim \mathsf{D}}\big(\rho(\tilde{X}, X) =1\big),$$
which represents the optimal a priori success rate only based on adversary's prior knowledge on $\mathsf{D}$ to recover $X$ such that the reconstruction satisfies the criterion defined by $\rho$.  
\label{def:prior_rate}
\end{Def}

\begin{Def}[$(\Delta_f\delta, \rho, \mathsf{D})$ PAC Advantage Privacy]
A mechanism $\mathcal{M}$ is termed $(\Delta_f\delta, \rho, \mathsf{D})$ PAC-advantage private if it is $(\delta, \rho, \mathsf{D})$ PAC private and 
$$ \Delta_f\delta \coloneqq  \mathcal{D}_{f}(\bm{1}_{\delta}\|\bm{1}_{\delta^{\rho}_o}) = \delta^{\rho}_of(\frac{\delta}{\delta^{\rho}_o}) + (1-\delta^{\rho}_o)f (\frac{1-\delta}{1-\delta^{\rho}_o}) .$$
Here, $\bm{1}_{\delta}$ and $\bm{1}_{\delta^{\rho}_o}$ represent two Bernoulli distributions of parameters $\delta$ and $\delta^{\rho}_o$, respectively.  \label{def:pac_advan}
\end{Def}

In Definition \ref{def:pac_advan},  $\Delta_f\delta$ captures the additional (posterior) advantage gained by the adversary after observing $\mathcal{M}(X)$. When we select $\mathcal{D}_f$ to be total variation, $\Delta_{TV}\delta$ simply becomes the difference $\delta^{\rho}_o-\delta$, as commonly used in classic cryptographic analysis. Accordingly, once we have an upper bound on $\delta^{\rho}_o-\delta$, plugging in $\delta^{\rho}_o$ provides a lower bound on the posterior failure rate $\delta$.

The seemingly artificial aspect of Definition \ref{def:pac_advan} in using $f$-divergence as the measure is motivated by Theorem \ref{thm: PAC_joint}, where we give generic upper bounds on $\Delta_f\delta$. In particular, when $\mathcal{D}_f$ is selected to be KL-divergence, we can simply bound $\Delta_{KL}\delta$ via mutual information. In addition, it is not hard to see that when $\mathcal{M}$ satisfies perfect secrecy, $\Delta_f\delta$ is $0$ for any $f$. Before the end of this section, we have the following important remarks on PAC Privacy. 

\noindent \textbf{Reconstruction Metric Selection}: Definitions \ref{def:pac}  and \ref{def:pac_advan} provide a unified framework to capture generic inference hardness. Membership and identification attacks are now special cases by selecting a specific data generation and a specific measure $\rho$. 
\begin{Def}[$(\delta, \rho, \mathsf{U}, \mathsf{D})$ PAC Membership Privacy]
For a data processing mechanism $\mathcal{M}$, given some measure $\rho$ and a data set $\mathsf{U}=(u_1, u_2, \ldots , u_{N})$, we say $\mathcal{M}$ satisfies $(\delta, \rho, \mathsf{U}, \mathsf{D})$-PAC Membership Privacy if the following experiment is impossible:

A user applies certain sampling $($described by $\mathsf{D})$ on $\mathsf{U}$ to generate a dataset $X \sim \mathsf{D}$, and sends $\mathcal{M}(X)$ to an adversary. The adversary who knows $\mathsf{U}, \mathsf{D}$ and $\mathcal{M}$ is asked to return an $N$-dimensional binary vector $ \tilde{\bm{1}}_{\mathsf{U}} = \big( \tilde{\bm{1}}_{u_1}, \ldots ,\tilde{\bm{1}}_{u_{N}}  \big)$ to predict the participation of each $u_i$, denote by ${\bm{1}}_{\mathsf{U}} = \big(\bm{1}_{u_1}, \ldots , \bm{1}_{u_N} \big)$, where $\bm{1}_{u_i}$ is an indicator to represent the participation of $u_i$, such that with probability at least $(1-\delta)$, $\rho(\tilde{\bm{1}}_{\mathsf{U}}, \bm{1}_{\mathsf{U}}) = 1$. 
\label{def:pac_membership}
\end{Def}

One can imagine the distribution $\mathsf{D}$ could be Poisson (i.i.d.) sampling or random sampling with(out) replacement from $\mathsf{U}$. From a DP angle, such privacy amplification from sampling has been studied in \cite{balle2018privacy}, \cite{wang2019subsampled}. PAC Membership Privacy is a very important variant of PAC Privacy in practical applications, which can give a formal security proof against the well-known and fundamental {\em membership inference attack} problem \cite{shokri2017membership}.

We may consider the classic statistical query model, where there is some sensitive data set $\mathsf{U}$, for example, medical records of $N$ patients, held by a server. The server is asked to answer some query, expressed as disclosing the results of $\mathcal{M}$. The server randomly samples a subset $X$ of elements from $\mathsf{U}$; for example, each datapoint in $\mathsf{U}$ will be independently selected with probability ${1}/{2}$, to participate in the computation of $\mathcal{M}(X)$. If we select $\rho(\tilde{\bm{1}}_{\mathsf{U}}, \bm{1}_{\mathsf{U}})=1$ only when $\tilde{\bm{1}}_{\mathsf{U}}$ and $\bm{1}_{\mathsf{U}}$ collide in at least one entry, then $(\delta, \rho, \mathsf{U}, \mathsf{D})$-PAC Membership Privacy suggests that {\em even if the adversary already knows the whole data pool $\mathsf{U}$, he still cannot successfully infer any participants in the processing mechanism $\mathcal{M}$}. { In other words, PAC Membership Privacy quantifies the {\em risk} incurred by any group of persons who allow their data to be processed through $\mathcal{M}$.}

If, on the other hand, our privacy concern pertains to individuals, we simply consider the reconstruction hardness of each single datapoint in the set $X$, and similarly define PAC Individual Privacy as follows.


\begin{Def}[$(\delta, \rho, n, \mathsf{D})$ PAC Individual Privacy]
For a data processing mechanism $\mathcal{M}$, given some data distribution $\mathsf{D}$, the total number $n$ of datapoints, a measure $\rho(\cdot, \cdot)$, we say $\mathcal{M}$ satisfies $(\delta, \rho, n, \mathsf{D})$-PAC Individual Privacy if for any $i \in [1:n]$ the following experiment is impossible:

A user generates data $X = (x_1, x_2, \ldots ,x_n)$ from distribution $\mathsf{D}$ and sends $\mathcal{M}(X)$ to an adversary. The adversary who knows $\mathsf{D}$ and $\mathcal{M}$ is asked to return an estimation $\tilde{x}_i$ on $x_i$ such that with probability at least $(1-\delta)$, $\rho(\tilde{x}_i, x_i) =1$. 
\label{def:pac_individual}
\end{Def}
If the input domain $\mathcal{X}$ is some finite set and we want to measure how hard it is for an adversary to identify the true selection of $x_i$, we may set $\rho$ to be $\rho(\tilde{x}_i, x_i)= 1$ iff $\tilde{x}_i = x_i$. Then, the reconstruction becomes the {\em identification problem}. We have to stress that a proper selection of measure $\rho$ is important to make PAC Privacy meaningful in different scenarios. Different error metrics will capture different perspectives of data privacy. For biological or financial data, such as ages and incomes, $l_p$ norm, such as $l_1$ and $l_2$, is usually meaningful. But for sensitive image or natural language data, $l_p$ norm may not be able to fully characterize the leakage of contextual information, and membership and identification hardness are better options. PAC Privacy provides a unified framework for arbitrary inference tasks and privacy concerns.

\noindent  \textbf{Lack of Knowledge}: In all the above definitions, we assume the adversary has full knowledge of the {\em a priori} data generation distribution $\mathsf{D}$, or $\mathsf{D}$ is public. However, our current PAC Privacy framework cannot properly quantify the inference hardness due to the {\em lack/absence} of prior knowledge. For example, suppose the data $X$ is generated from some distribution $\mathsf{D}$, say a random picture selected from a pool of cat images. For an adversary who does not have full knowledge but only knows that the users' pictures could be selected from an image pool of cats or dogs, the inference would be harder compared to the full knowledge setup. Unfortunately, our techniques used in PAC Privacy rely on the entropy of data, which does not change conditional on such lack of knowledge. However, in the other direction, PAC Privacy {\em does} quantify the inference hardness in a more generic scenario when the adversary has more knowledge rather than simply knowing $\mathsf{D}$.  We refer the reader to Remark \ref{rmk_prior/baised} in Section \ref{sec:bound}.

Although our current framework cannot exactly quantify the inference hardness amplification from the lack of knowledge, PAC Privacy enjoys strong and realistic privacy interpretation from this perspective. If a certain inference problem is already difficult for an adversary who has full knowledge, then it is even more difficult for an adversary who lacks knowledge.
Suppose a user holds a set $\mathsf{U}$ of cat images and randomly samples, say, 1000 samples from them, to produce $X$ and discloses $\mathcal{M}(X)$. If $\mathcal{M}(X)$ satisfies PAC Privacy such that even an adversary who knows $\mathsf{U}$ and the data generation $\mathsf{D}$ cannot identify or recover the 1000 participants, then it is even more difficult for an adversary to identify or recover the 1000 selected samples from the universe of cat images. 

\section{Foundation of PAC Privacy}
\label{sec:bound}
In this section, we present our generic results to enable concrete PAC Privacy analysis. Before we state the theorems, we need the following additional notations to express the results. 

\begin{Def}[Entropy]
For a random variable $x$ defined on a discrete set $\mathcal{X}$, the (Shannon) entropy $\mathcal{H}(x)$ of $x$ is defined as $\mathcal{H}(x) = \sum_{x_0 \in \mathcal{X}} -\Pr(x=x_0)\log \Pr(x=x_0).$ If $x$ is in a continuous distribution on some domain $\mathcal{X}$, then the (differential) entropy $h(x)$ of $x$ is defined as $h(x) = \int_{x_0 \in \mathcal{X}} -\mathbb{P}(x=x_0)\log \mathbb{P}(x=x_0) d x_0,$ where $\mathbb{P}$ represents the probability density function. 
\end{Def}
With a slight abuse of notation, we simply apply $\mathcal{H}(x)$ to denote the (Shannon/differential) entropy of $x$ depending on its distribution for simplicity. Accordingly, we can define the conditional entropy $\mathcal{H}(x|w)$ for two random variables $x$ and $w$ in some joint distribution, where $\mathcal{H}(x|w) = \sum_{w_0 \in \mathcal{W}} \mathcal{H}(x|w=w_0)\Pr(w=w_0).$ When $w$ is in some continuous distribution, we can similarly define the conditional entropy in an integral form as differential entropy. In the following, we will use $\mathsf{P}_x$ to denote the distribution of a random variable $x$. 

\begin{Def}[Mutual Information]
For two random variables $x$ and $w$ in some joint distribution, the mutual information $\mathsf{MI}(x;w)$ is defined as $$\mathsf{MI}(x;w) \coloneqq  \mathcal{H}(x)-\mathcal{H}(x|w)= \mathcal{D}_{KL}(\mathsf{P}_{x,w}\|\mathsf{P}_{x} \otimes \mathsf{P}_{w} ),$$
i.e., the KL-divergence between the joint distribution of $(x,w)$ and the product of the marginal distributions of $x$ and $w$, respectively. 
\end{Def}

Mutual information measures the dependence between $x$ and $w$. $\mathsf{MI}(x;w) = \mathsf{MI}(w;x)$ is always non-negative and equals $0$ if $x$ is independent of $w$.

\begin{Thm}
\label{thm: PAC_joint}
For any selected $f$-divergence $\mathcal{D}_{f}$, a mechanism $\mathcal{M}: \mathcal{X}^* \to \mathcal{Y}$ satisfies $(\Delta_{f}\delta, \rho, \mathsf{D})$ {PAC Advantage Privacy} if  
\begin{equation}
\label{generic_fano_Df}
    \begin{aligned}
    \Delta_f\delta = \mathcal{D}_{f} \big( \bm{1}_{\delta} \| \bm{1}_{\delta^{\rho}_o} \big) \leq \inf_{\mathsf{P}_{W}} \mathcal{D}_{f}\big( \mathsf{P}_{X, \mathcal{M}(X)} \| \mathsf{P}_{X} \otimes \mathsf{P}_{W} \big).
    \end{aligned}
\end{equation}
Here, $\bm{1}_{\delta}$ and $\bm{1}_{\delta^{\rho}_o}$ are two Bernoulli distributions of parameters $\delta$ and $\delta^{\rho}_o$, respectively; $\mathsf{P}_{X, \mathcal{M}(X)}$ and $\mathsf{P}_{X}$ are the joint distribution and the marginal distribution of $(X, \mathcal{M}(X))$ and $X$, respectively; and $\mathsf{P}_{W}$ represents the distribution of an arbitrary random variable $W \in \mathcal{Y}$. 

In particular, when we select $\mathcal{D}_{f}$ to be the KL-divergence and $\mathcal{P}_W = \mathcal{P}_{\mathcal{M}(X)}$, $\mathcal{M}$ satisfies $(\Delta_{KL}\delta, \rho, \mathsf{D})$ PAC Advantage Privacy where 
\begin{equation}
\label{generic_fano}
    \begin{aligned}
    \Delta_{KL}\delta \coloneqq  \mathcal{D}_{KL}(\bm{1}_{\delta}\|\bm{1}_{\delta^{\rho}_o}) \leq \mathsf{MI}\big(X;\mathcal{M}(X)\big).
    \end{aligned}
\end{equation}
Here, $\mathcal{D}_{KL}(\bm{1}_{\delta}\|\bm{1}_{\delta^{\rho}_o}) = \delta\log(\frac{\delta}{\delta^{\rho}_o}) + (1-\delta)\log (\frac{1-\delta}{1-\delta^{\rho}_o})$. 
\end{Thm}

\begin{proof}
See Appendix \ref{app:thm:Pac_joint}.
\end{proof}

\begin{Rmk}
\label{rmk_prior/baised}
In general, if the adversary has more prior knowledge $($besides the data distribution $\mathsf{D})$ on $X$, denoted by a variable $Adv$, we can obtain a corresponding PAC Privacy bound using Theorem \ref{thm: PAC_joint}, where the only difference is that in $(\ref{generic_fano})$, $\mathsf{MI}(X; \mathcal{M}(X))$ and $\delta^{\rho}_o$ are changed to the conditional $\mathsf{MI}(X; \mathcal{M}(X) |Adv)$ and the corresponding failure probability of optimal {\rm a priori} guessing on $X$ conditional on $Adv$, respectively.
\end{Rmk}

Theorem \ref{thm: PAC_joint} is rooted in the well-known data processing inequality of $f$-divergence. 
Roughly speaking, this inequality says postprocessing cannot increase information. Given priors, the adversary's reconstruction or decision is determined by (a postprocessing of) his observation on $\mathcal{M}(X)$. Therefore, to information-theoretically bound the posterior advantage, it suffices to quantify the correlation between $X$ and $\mathcal{M}(X)$. Theorem \ref{thm: PAC_joint} serves as the foundation of PAC Privacy analysis. From (\ref{generic_fano_Df}), we know that $\Delta_f\delta=\mathcal{D}_{f}(\bm{1}_{\delta}\|\bm{1}_{\delta^{\rho}_o})$, the change to the $f$-divergence between the optimal {\em a priori} and posterior reconstruction/inference performance, is bounded by the corresponding $f$-divergence between the distributions of $(X,\mathcal{M}(X))$ and $(X, W)$, where $W$ is an arbitrary random variable independent of $X$. With the freedom of the selections of both $f$ and $W$, if we apply KL-divergence and  set $W$ to be distributed the same as ${\mathcal{M}(X)}$, we have $\Delta_{KL}\delta \leq \mathsf{MI}(X; \mathcal{M}(X)),$ the mutual information between $X$ and $\mathcal{M}(X)$.  

It is noted that we are also free to select {\em arbitrary} $\rho$ as the objective inference task and the success criterion in (\ref{generic_fano_Df}). Thus, another important implication from Theorem \ref{alg: deter_MI} is that {once we have an upper bound on $\mathcal{D}_{f}\big( \mathsf{P}_{X, \mathcal{M}(X)} \| \mathsf{P}_{X} \otimes \mathsf{P}_{\mathcal{M}(X)} \big)$}, we indeed also control {\em the posterior advantage of all possible adversaries' inference tasks.} This is also the foundation of the automatic proof of PAC Privacy, where the {\em 
entirety of data generation and processing $\mathcal{M}$ could be a black-box oracle}; as shown in Section \ref{sec:MI}, we are able to control such $f$-divergence without any algorithmic analysis of $\mathcal{M}$. We can compute $\delta^{\rho}_o$ based on $\rho$ and $\mathsf{D}$ for any given inference task. Given $\delta^{\rho}_o$, once we have the upper bound on $\Delta_{f}\delta$, we can consequently lower bound the adversary's failure probability $\delta$ for inference since $\mathcal{D}_f$ is a deterministic function.  

In the following, we will mainly focus on $\Delta_{KL}\delta$ to control PAC Privacy since mutual information is a well-studied quantity in information theory and has many nice properties to simplify our analysis. Intuitively, $\mathsf{MI}(X;\mathcal{M}(X))$ quantifies the dependence between $X$ and $\mathcal{M}(X)$, and shows how much such instance-based information (in nats/bits by taking $\log_e$/$\log_2$ in KL-divergence) is provided by $\mathcal{M}(X)$ to support the adversary's inference. Especially, when we select $\mathcal{D}_f$ to be total variation and still $\mathsf{P}_W = \mathsf{P}_{\mathcal{M}(X)}$, 
$$ {\Delta}_{TV}{\delta} = \delta^{\rho}_{o} - \delta \leq \mathcal{D}_{TV} \big( \mathsf{P}_{X,\mathcal{M}(X)} \| \mathsf{P}_{X} \otimes \mathsf{P}_{\mathcal{M}(X)} \big).$$
By Pinsker's inequality, which states that for two distributions $\mathsf{P}$ and $\mathsf{Q}$ on the same measurable space, $ {\Delta}_{TV}(\mathsf{P} \| \mathsf{Q}) \leq \sqrt{\frac{1}{2} \cdot {\Delta}_{KL}(\mathsf{P} \| \mathsf{Q})},$
we have the following corollary to give a quantification on the more intuitive and classic posterior advantage. 
\begin{Cor}
With the same setup as Theorem \ref{thm: PAC_joint}, the posterior advantage $(\delta^{\rho}_o - \delta)$ for arbitrary $\rho$ is upper bounded as  
$$ \Delta_{TV}\delta =  \delta^{\rho}_o - \delta \leq \sqrt{\frac{1}{2} \cdot \mathsf{MI}\big(X; \mathcal{M}(X)\big)}.$$
\label{cor:KL-tv}
\end{Cor}

On the other hand, we can also apply Theorem \ref{thm: PAC_joint} to interpret the classic worst-case security definitions. To derive the {\em data-independent} maximal leakage when exposing $\mathcal{M}(X)$, we may take a supremum on the right side of (\ref{generic_fano}) and study $\sup_{\mathsf{D}, X \sim \mathsf{D}}\mathsf{MI}(X; \mathcal{M}(X))$ instead. Thus, the adversarial worst-case analysis can be used to produce a (loose) upper bound of PAC Privacy. For example, if we know $\mathcal{M}$ is $\epsilon$-DP, then $\sup_{\mathsf{D}, X \sim \mathsf{D}}\mathsf{MI}(X; \mathcal{M}(X)) \leq 0.5\epsilon^2n^2$ for an $n$-datapoint set $X$ (Theorem 1.10 and Proposition 1.4 in \cite{bun2016concentrated}). 

With Theorem  \ref{thm: PAC_joint}, our goal for the remainder of this paper is clear, where we set out to determine or bound $\mathsf{MI}(X;\mathcal{M}(X))$ with high confidence. We first move on to analyze an important special case where $x_i$ is i.i.d. generated and $\mathcal{M}(\cdot)$ is a symmetric mechanism. Here, we call $\mathcal{M}(x_1, x_2, \ldots , x_n)$ symmetric if for an arbitrary perturbation $\pi$ on $[1:n]$, the distribution of $\mathcal{M}(x_{\pi(1)}, x_{\pi(2)}, \ldots , x_{\pi(n)})$ is identical, where each element $x_i$ is treated equally. Many data processing tasks satisfy the symmetric property, for example mean/median estimation or empirical risk minimization (ERM) in machine learning, where the order of samples does not change the result. The main motivation to study this special case is that we can show tighter and simpler bounds to control the security parameter $\delta$ in many applications, especially when the mechanism $\mathcal{M}$ is deterministic with respect to $X$. This will become clear in Section \ref{sec:MI}. For notional clarity, we will use $X \sim \bar{\mathsf{D}}^n$ to denote that $x_i$ is i.i.d. generated from some distribution $\bar{\mathsf{D}}$ to differentiate from the generic joint distribution case. 

\begin{Thm}
A mechanism $\mathcal{M}$ satisfies $(\delta, \rho, n, \mathsf{D})$ PAC Individual Privacy if there exist $\delta_i$ for $i=1,2,\ldots,n,$ satisfying $\delta_i \geq \delta$ and 
\begin{equation}
    \begin{aligned}
    \mathcal{D}_{KL}(\bm{1}_{\delta_i}\|\bm{1}_{\delta^{\rho}_{o,i}}) \leq \mathsf{MI}\big(x_i;\mathcal{M}(X)\big).
    \end{aligned}
    \label{generic_fano_individual}
\end{equation}
Here, $1-\delta^{\rho}_{o,i}  = \sup_{\tilde{x}_i \in \mathcal{X}}\Pr( \rho( \tilde{x}_i, x_i) = 1)$ represents the optimal prior success rate of recovering $x_i$. 

In particular, when $X \sim \bar{\mathsf{D}}^n$ and $\mathcal{M}$ is a symmetric mechanism, then 
$$ \delta \geq \sum_{j=1}^n \frac{\tilde{\delta}^{j}}{n},$$
for any selection of $\tilde{\delta}^{j}$ satisfying 
$\mathcal{D}_{KL}(\bm{1}_{\tilde{\delta}^j}\|\bm{1}_{\tilde{\delta}^{j, \rho}_{o}}) \leq \mathsf{MI}\big(X;\mathcal{M}(X)\big),$
where 
\begin{equation}
1-\tilde{\delta}^{j, \rho}_{o}= \sum_{l=j}^n {n \choose l} (\bar{\delta}^{\rho}_{o})^{n-l} \cdot (1-\bar{\delta}^{\rho}_{o})^{l},
\label{tilde_delta}
\end{equation}
and $1-\bar{\delta}^{\rho}_{o}  = \sup_{\tilde{x} \in \mathcal{X}}\Pr_{x \sim \bar{D}}({\rho}(\tilde{x}, x)=1).$
\label{thm:i.i.d.pac}
\end{Thm}
\begin{proof}
See Appendix \ref{app:thm:i.i.d.pac}. 
\end{proof}

In Theorem \ref{thm:i.i.d.pac}, we point out that if our objective is individual privacy, we may get tighter bounds by applying $\mathsf{MI}(x_i;\mathcal{M}(X))$ rather than $\mathsf{MI}(X;\mathcal{M}(X))$. By the chain rule of mutual information \cite{cover1999elements}, we know that $\mathsf{MI}(X;\mathcal{M}(X))$ is an upper bound for any $\mathsf{MI}(x_i;\mathcal{M}(X))$ since for any $i$, 
$$\mathsf{MI}(X;\mathcal{M}(X)) =\mathsf{MI}(x_i;\mathcal{M}(X)) +\mathsf{MI}(X\backslash\{x_i\};\mathcal{M}(X)|x_i)\geq \mathsf{MI}(x_i;\mathcal{M}(X)).$$  However, to analyze $\mathsf{MI}(x_i;\mathcal{M}(X))$, generally we need to take the other elements $X\backslash\{x_i\}$ as part of randomness seeds and enforce $\mathcal{M}$ to be randomized, which makes it more difficult to give tight bounds. Therefore, in the second part of Theorem \ref{thm:i.i.d.pac}, in the case of i.i.d. data with symmetric processing, we give a more fine-grained analysis but stick to $\mathsf{MI}(X,\mathcal{M}(X))$. The idea behind this is that we consider the chance, denoted by $(1-\tilde{\delta}^j)$, that the adversary can successfully recover at least $j$ many elements out of an $n$-sample set $X$, and use the symmetric property to bound the success rate of each individual $x_i$.

\begin{Ex}[Applications of Theorems \ref{thm: PAC_joint} and \ref{thm:i.i.d.pac}] We consider a set $X$ of $n$ datapoints where each datapoint $x_i$ is i.i.d. uniformly selected from a finite set of size $N$ and $\mathcal{M}$ is some symmetric mechanism. Suppose $N = 100$ and for any given $i$, we know the optimal prior success rate $(1-\delta^{\rho}_o)$ to identify the true selection of $x_i$ is $0.01$. {When $\mathsf{MI}(X; \mathcal{M}(X))=1$, via Theorem \ref{thm: PAC_joint}, by numerically solving $ \mathcal{D}_{KL}(\bm{1}_{\delta_i}\|\bm{1}_{0.99}) \leq 1$, we have a global bound $(1-\delta) \leq 0.36$ for any $n$. 

Applying Theorem \ref{thm:i.i.d.pac}, a larger $n$ will produce a tighter bound of $(1-\delta)$. Given $\mathsf{MI}(X; \mathcal{M}(X))=1$, we numerically solve $\mathcal{D}_{KL}(\bm{1}_{\tilde{\delta}^j}\|\bm{1}_{\tilde{\delta}^{j, \rho}_{o}}) \leq 1$ to obtain the lower bound of $\tilde{\delta}_j$, the failure probability that an adversary can recover at least $j$ elements for $j=1,2,\cdots, n$, and  consequently obtain the lower bound of individual inference failure probability $\delta$. When $n=10$, we have $(1-\delta) \leq 0.17,$ and when $n=50$, we have $(1-\delta) \leq 0.06$. }
\end{Ex}

So far, we have presented the generic bounds of PAC Privacy, and to ensure desired security parameters, the problem is reduced to controlling the mutual information $\mathsf{MI}(X; \mathcal{M}(X))$ \big(or $\mathsf{MI}(x_i; \mathcal{M}(X))$\big) with high confidence. 

\section{Automatic Control of Mutual Information}
\label{sec:MI}
In this section, we present our main results on automatic privacy measurement and control. At a high level, we want an automatic protocol which takes security parameters as input and returns a privatization scheme on $\mathcal{M}$ to ensure required privacy guarantees with high confidence. In particular, we want the least assumptions on mechanism $\mathcal{M}$, which ideally could be a black-box oracle and no specific algorithmic analysis is needed to produce the security proof. 

One natural idea to achieve information leakage control is perturbation: when $\mathsf{MI}(X; \mathcal{M}(X))$ is not small enough to produce satisfied PAC Privacy, we may add additional noise $\bm{B}$, say Gaussian, to produce smaller $\mathsf{MI}(X; \mathcal{M}(X)+\bm{B})$. In PAC Privacy, the role of noise is not simply perturbation but to also enforce the output of a black-box oracle into a class of parameterized distributions. As shown below for either deterministic or randomized algorithms, with Gaussian noise, the analysis is reduced to the study of divergences of Gaussian mixture models. The key question left is how to automatically determine the parameters of the noise $\bm{B}$. We give solutions in this section. In Section \ref{sec:deter_alg}, we focus on the deterministic $\mathcal{M}(\cdot)$ w.r.t. $X$, while, in Section \ref{sec:random_alg}, we analyze generic randomized $\mathcal{M}(\cdot)$ w.r.t. either $X$ or a single datapoint $x_i$. Finally, in Section \ref{sec:optimal}, we discuss how to approximate the optimal perturbation. 

\subsection{Deterministic Mechanism}
\label{sec:deter_alg}
The estimation of mutual information is a long-standing open problem, though it enjoys a very simple expression via (conditional) entropy or KL-divergence. Practical mutual information estimations have been studied in many empirical works using deep learning \cite{belghazi2018mine}, \cite{cheng2020club} or kernel methods \cite{kernel} with different motivations such as information bottleneck \cite{information_bottleneck} and causality \cite{butte1999mutual}. However, one key obstacle to use those heuristic estimations is that we cannot derive a high confidence bound and design necessary perturbation for rigorous security guarantees. 

In this section, we assume that $\mathcal{M}$ is some deterministic function w.r.t. $X$, for example, $\mathcal{M}$ returns the average/median of $X$ or the global optimum of some loss function determined by $X$. Beyond the deterministic assumption, the only other restriction is that the output $\mathcal{M}(X)$ is bounded. Without loss of generality, we assume the $l_2$ norm $\|\mathcal{M}(x)\|_{2} \leq r$. Since the distribution of $\mathcal{M}(X)$ could be arbitrary over $\mathbb{B}^d_{r}$, the $d$-dimensional $l_2$ ball of radius $r$ centered at zero, in general, there do not exist non-trivial upper bounds of $\mathsf{MI}(X; \mathcal{M}(X))$. Indeed, possibly counter-intuitively, we cannot even say that the aggregation of a larger sample set would incur less information leakage. For example, imagine a finite set $\mathsf{U}$ where the empirical means of its subsets are distinct. $X$ represents a subset of $\mathsf{U}$ randomly selected and $\mathcal{M}(X)$ outputs the average of $X$. Since there is a bijection between $X$ and $\mathcal{M}(X)$, $\mathsf{MI}(X; \mathcal{M}(X)) = \mathcal{H}(X) = O(\log{2^{|\mathsf{U}|}})$, increasing with $|\mathsf{U}|$. 

Fortunately, adding {\em continuous} noise, especially Gaussian noise,  enables us to derive generic control of $\mathsf{MI}(X; \mathcal{M}(X)+\bm{B})$. Continuous noise $\bm{B}$ enforces the distribution of output $\mathcal{M}(X)+\bm{B}$ to be parameterized and continuous regardless of $\mathcal{M}$ and $X$. Below, we show how to apply covariance estimation to fully automatically determine a noise $\bm{B}$ to control $\mathsf{MI}(X; \mathcal{M}(X)+\bm{B})$ with confidence for any black-box oracle $\mathcal{M}$. In the following, we simply set $\bm{B}$ to be some multivariate Gaussian in a form $\mathcal{N}(0, \Sigma_{\bm{B}})$, whose covariance is $\Sigma_{\bm{B}}$. Similarly, $\Sigma_{\mathcal{M}(X)}$ represents the covariance matrix of $\mathcal{M}(X)$. 

\begin{Thm} For an arbitrary deterministic mechanism $\mathcal{M}$ and a Gaussian noise of the form $\bm{B} \sim \mathcal{N}(\bm{0}, \Sigma_{\bm{B}})$,  
\begin{equation}
 \mathsf{MI} (X; \mathcal{M}(X)+\bm{B}) \leq \frac{1}{2} \cdot \log \text{det}\big(\bm{I}_{d}+ \Sigma_{\mathcal{M}(X)}\cdot \Sigma^{-1}_{\bm{B}}\big).
\end{equation}
In particular, let the eigenvalues of $\Sigma_{\mathcal{M}(X)}$ be  $(\lambda_1, \ldots, \lambda_{d})$, then there exists some $\Sigma_{\bm{B}}$ such that $\mathbb{E}[\|\bm{B}\|^2_2] = (\sum_{j=1}^{d}\sqrt{\lambda_j})^2,$ and  $\mathsf{MI} (X; \mathcal{M}(X)+\bm{B}) \leq 1/2.$
\label{thm:deter_m}
\end{Thm}

\begin{proof}
See Appendix \ref{app:thm:deter_m}. 
\end{proof}

From Theorem \ref{thm:deter_m}, we have a simple upper bound on the mutual information after perturbation which only requires the knowledge of the covariance of $\mathcal{M}(X)$ when $\mathcal{M}$ is deterministic. Another important and appealing property is that the noise calibrated to ensure the target mutual information bound is {\em not} explicitly dependent on the output dimensionality $d$ but instead on the square root sum of eigenvalues of $\Sigma_{\mathcal{M}(X)}$. When $\mathcal{M}(X)$ is largely distributed in a $p$-rank subspace of $\mathbb{R}^d$, Theorem \ref{thm:deter_m} suggests that a noise of scale $O(\sqrt{p})$ is needed. This is different from DP where the expected $l_2$ norm of noise is in a scale of ${\Theta}(\sqrt{d})$ given constant $l_2$-norm sensitivity; we will present a comprehensive comparison in Section \ref{sec:worst_vs_average}.  

Based on Theorem \ref{thm:deter_m}, we now proceed to present an automatic protocol Algorithm \ref{alg: deter_MI} to determine $\Sigma_{\bm{B}}$ and produce an upper bound such that $\mathsf{MI} (X; \mathcal{M}(X)+\bm{B}) \leq (v+\beta)$ with high confidence, where $v$ and $\beta$ are positive parameters selected as explained below. After sufficiently many simulations,  the following theorem ensures that we can obtain accurate enough estimation with arbitrarily high probability. A characterization on the relationship among security parameter $c$, simulation complexity $m$ and confidence parameter $\gamma$ is given in (\ref{tradeoff_deter_m}).  

\begin{algorithm}[t]
\caption{$(1-\gamma)$-Confidence Noise Determination of Deterministic Mechanism $\mathcal{M}$}
\begin{algorithmic}[1]
\STATE \textbf{Input:} A deterministic mechanism $\mathcal{M}: \mathcal{X}^* \to \mathbb{B}^d_{r}$, data distribution $\mathsf{D}$, sampling complexity $m$, security parameter $c$, and mutual information quantities $v$ and $\beta$.  
\FOR{$k=1,2,\ldots,m$}
\STATE Independently generate data $X^{(k)}$ from distribution $\mathsf{D}$.  
\STATE Record $y^{(k)}= \mathcal{M}(X^{(k)})$.
\ENDFOR 
\STATE Calculate empirical mean $\hat{\mu}={\sum_{k=1}^m y^{(k)}}/{m}$ and the empirical covariance estimation
$\hat{\Sigma} = {\sum_{k=1}^m (y^{(k)}-\hat{\mu})(y^{(k)}-\hat{\mu})^T}/{m}.$
\STATE  Apply singular value decomposition (SVD) on $\hat{\Sigma}$ and obtain the decomposition as $\hat{\Sigma}=\hat{U}\hat{\Lambda} \hat{U}^T,$ where $\hat{\Lambda}$ is the diagonal matrix of eigenvalues $\hat{\lambda}_1 \geq \hat{\lambda}_2 \geq \ldots \geq  \hat{\lambda}_d$.   
\STATE Determine the maximal index $j_0 = \arg \max_j \hat{\lambda}_j$, for those $\hat{\lambda}_j > c$.
\IF{$\min_{ 1\leq j \leq j_0, 1 \leq l \leq d}  |\hat{\lambda}_j-\hat{\lambda}_l| > r\sqrt{dc}+2c$}
\FOR{$j=1, 2, \ldots, d$}
\STATE Determine the $j$-th element of a diagonal matrix $\Lambda_B$ as $$\lambda_{B,j} = \frac{2v}{ \sqrt{\hat{\lambda}_j+10cv/\beta} \cdot \big(\sum_{j=1}^d \sqrt{\hat{\lambda}_j+10cv/\beta} \big)}.$$
\ENDFOR
\STATE Determine the Gaussian noise covariance as $\Sigma_{\bm{B}} = \hat{U}\Lambda^{-1}_{\bm{B}}\hat{U}^T$.
\ELSE
\STATE Determine the Gaussian noise covariance as $\Sigma_{\bm{B}}= (\sum_{j=1}^d \hat{\lambda}_j+dc)/(2v) \cdot \bm{I}_d$.
\ENDIF 
\STATE \textbf{Output}: Gaussian covariance matrix $\Sigma_{\bm{B}}$.  
\end{algorithmic}
\label{alg: deter_MI}
\end{algorithm}

\begin{Thm}
Assume that $\mathcal{M}(X) \in \mathbb{R}^d$ and $
\|\mathcal{M}(X)\|_2 \leq r$ for some constant $r$ uniformly for any $X$, and apply Algorithm \ref{alg: deter_MI} to obtain the Gaussian noise covariance $\Sigma_{\bm{B}}$ for a specified mutual information bound $v + \beta$. $v$ and $\beta$ can be chosen independently, and $c$ is a security parameter. Then, there exists a fixed and universal constant $\kappa$ such that one can ensure $\mathsf{MI}(X;\mathcal{M}(X)+\bm{B}) \leq v+\beta$ with confidence at least $(1-\gamma)$ once the selections of $c$, $m$ and $\gamma$ satisfy, 
\begin{equation}
    \label{tradeoff_deter_m}
    c \geq \kappa r\big(\max\{ \sqrt{\frac{d+\log(4/\gamma)}{m}},\frac{d+\log(4/\gamma)}{m}\} + \sqrt{\frac{d\log(4/\gamma)}{m}} \big).
\end{equation}
\label{thm: deter_alg}
\end{Thm}

\begin{proof}
See Appendix \ref{app:proof_deter_alg}.
\end{proof}

At a high level, Algorithm \ref{alg: deter_MI} shows a way to estimate both the eigenvectors and the spectrum (eigenvalues) of the covariance matrix $\Sigma_{\mathcal{M}(X)}$ (lines 6-7), which afterwards determines the noise parameter $\Sigma_{\bm{B}}$ (lines 8-16) to ensure desired mutual information upper bound $(v+\beta)$. In the first part, we simply apply the eigenvalue and eigenvectors of the empirical covariance obtained by the results of $m$ simulations as the estimations of those of the true covariance matrix $\Sigma_{\mathcal{M}(X)}$. In the second part to determine the noise $\bm{B}$, depending on the eigenvalue gap (line 9), we present two different upper bounds of ${\Sigma}_{\bm{B}}$ as explained later. 

Theorem \ref{thm: deter_alg} presents a generic tradeoff between noise, sampling complexity and confidence. The role of $c$ in Algorithm \ref{alg: deter_MI} is a safe parameter that provides a lower bound on noise. From (\ref{tradeoff_deter_m}), in general, we may use a larger noise (larger $c$ for looser noise estimation) to produce higher confidence (smaller $\gamma$) given limited computation power $m$ for sampling and simulation; on the other hand, both $c$ and $\beta$ can also be arbitrarily small provided sufficiently many simulations $m$ to produce accurate enough estimation. 
 
The main reason for the need to estimate the eigenvectors of $\Sigma_{\mathcal{M}(X)}$ in Algorithm \ref{alg: deter_MI} is because we need the instance-based noise to fit the geometry of the eigenspace of $\Sigma_{\mathcal{M}(X)}$. This is different from the adversarial worst-case, such as DP, which only focuses on the magnitude of sensitivity and adds isometric noise; {\em in PAC Privacy, we add anisotropic noise, as much as needed in each direction}. The key intuition behind why we can ensure arbitrarily high confidence is because the output domain of $\mathcal{M}(X)$ is bounded and thus $\mathcal{M}(X)$ is in a certain sub-Gaussian distribution with concentration guarantees captured by bounded $r$. Therefore, provided sufficiently many simulations, we can approximate the true covariance matrix $\Sigma_{\mathcal{M}(X)}$ with arbitrary precision with a mild assumption on the eigenvalue gap. 

It is noted that depending on the {\em eigenvalue gap} defined in line 9 of Algorithm \ref{alg: deter_MI}, there are two cases of perturbation. When the {\em significant eigenvalues} of $\Sigma_{\mathcal{M}(X)}$ are distinct, then for arbitrary high confidence $(1-\gamma)$, the determined noise covariance $\Sigma_{\bm{B}}$ matches that described in Theorem \ref{thm:deter_m} given large enough $m$, which allows selection of small enough $c$. However, though the spectrum of $\Sigma_{\mathcal{M}(X)}$ can be estimated with arbitrary high precision, we cannot ensure a satisfied eigenvector space approximation in our current analysis framework without the assumption on the eigenvalue gap. To this end, in the special case without the eigenvalue gap guarantee, we have to switch to a more pessimistic perturbation (described in line 15 in Algorithm \ref{alg: deter_MI}), where the $l_2$ norm of the noise is then explicitly dimension dependent, ${\Theta}(\sqrt{d})$. In practice, such a case is very rare, and we leave the improvement of Algorithm \ref{alg: deter_MI} or a tighter analysis as an open question.

\subsection{Randomized Mechanism}
\label{sec:random_alg}
We proceed to consider the more complicated case where the mechanism $\mathcal{M}$ is randomized, where to be specific we use the form $\mathcal{M}(X,\theta)$ and $\theta$ is the randomness seed. Given selection of $\theta$, $\mathcal{M}(\cdot, \theta)$ becomes deterministic. As mentioned earlier, to analyze $\mathsf{MI}(x_i;\mathcal{M}(X))$, we can simply combine $X\backslash \{x_i\}$, the other elements in $X$ without $x_i$, with $\theta$ as randomness seeds. Therefore, without loss of generality, we only focus on the control of $\mathsf{MI}(X;\mathcal{M}(X)+\bm{B})$. The randomness seed $\theta$ of mechanism $\mathcal{M}$ is randomly selected from a set $\bm{\Theta} = \{\theta_1, \theta_2, \ldots , \theta_{|\bm{\Theta}|}\}$. As a simple generalization of Theorem \ref{thm:deter_m}, we have the following bound.  
\begin{Cor}
\begin{equation}
\label{mi_cor}
 \mathsf{MI} (X; \mathcal{M}(X, \theta)+\bm{B}(\theta)) \leq \mathbb{E}_{\theta} \frac{\log \text{det}(\bm{I}_{d}+ \Sigma_{\mathcal{M}(X, \theta)}\cdot \Sigma^{-1}_{\bm{B}(\theta)})}{2}.
\end{equation}
Here, we assume a random-seed dependent Gaussian noise $\bm{B}(\theta)$, where when $\theta = \theta_{0}$, $\bm{B} \sim \mathcal{N}(0, \Sigma_{\bm{B}(\theta_0)}).$
\label{cor:random_m}
\end{Cor}

Corollary \ref{cor:random_m} utilizes the following properties of mutual information.
\begin{align*}
   \mathsf{MI}(\theta, \mathcal{M}(X,\theta); X) = \mathsf{MI}(\theta;X) + \mathsf{MI}(\mathcal{M}(X, \theta);X|\theta) =  \mathsf{MI}(\mathcal{M}(X, \theta);X|\theta);
\end{align*}
$$\mathsf{MI}(\theta, \mathcal{M}(X,\theta);X) =\mathsf{MI}(\mathcal{M}(X,\theta);X)+\mathsf{MI}(\theta;X|\mathcal{M}(X; \theta)) \geq \mathsf{MI}(\mathcal{M}(X,\theta);X). $$
Therefore, $\mathsf{MI} (X; \mathcal{M}(X, \theta)+\bm{B}(\theta))$ is upper bounded by $ \mathsf{MI}(X; \mathcal{M}(X, \theta)+\bm{B}(\theta)|\theta)$, where given $\theta$, $\mathcal{M}$ is deterministic. Thus, we can apply the results of Section \ref{sec:deter_alg} to control each subcase dependent on the selection of $\theta$, and Corollary \ref{cor:random_m} follows straightforwardly. However,  there are two potential limitations of Corollary \ref{cor:random_m}. First, there could be exponentially many possible selections of $\theta$ and estimations on all $\Sigma_{\mathcal{M}(X, \theta)}$ could be intractable. The other issue is that (\ref{mi_cor}) is tight only if for given $\theta$, $\mathcal{M}(X,\theta)$ is concentrated. However, in some applications this may not be true. One example is non-convex optimization, e.g., we run SGD in deep learning. If the generalization error is small, the global distributions of $\mathcal{M}(X, \theta)$ and $\mathcal{M}(X', \theta)$ for any two representative datasets $X$ and $X'$ could be very close, but different $X$ with different randomness may converge to very different local minima.  Therefore, we propose an alternative approach as Algorithm \ref{alg: random_MI} and its security proof is shown in Theorem \ref{thm: random_alg}. 

In Algorithm \ref{alg: random_MI}, we need the following definition of {\em minimal-permutation distance}. The minimal-permutation distance between two $k$-block vectors $\bm{a}$ and $\bm{b}$ is defined as $\mathsf{d}_{\pi}(\bm{a},\bm{b}) = \min_{\pi} \sum_{j} \|\bm{a}(j)-\bm{b}(\pi(j))\|^2/k,$ where $\bm{a}(j)$ and $\bm{b}(j)$ represent their $j$-th blocks, respectively, and $\pi$ is some permutation on the block index $[1:k]$.

\begin{algorithm}[t]
\caption{$(1-\gamma)$-Confidence Noise Determination of Randomized Mechanism $\mathcal{M}(\cdot, \theta)$}
\begin{algorithmic}[1]
\STATE \textbf{Input:} A randomized mechanism $\mathcal{M}(\cdot, \theta): \mathcal{X}^* \to \mathbb{B}^d_{r}$, $\theta \in \bm{\Theta}$, data distribution $\mathsf{D}$, sampling complexity $m$, mutual information bound $v$, and security parameters $c$ and $\tau$. 
\FOR{$k=1,2,\ldots,m$}
\STATE From distribution $\mathcal{D}$, independently sample ${X}^{(k)}_1$ and
${X}^{(k)}_2.$
\STATE Randomly select $\tau$ seeds, denoted by ${\mathcal{C}}^{(k)} = \{\theta^{(k)}_1, \theta^{(k)}_2, \ldots , \theta^{(k)}_\tau\}$, from $\bm{\Theta}$. 
\STATE Compute and record $$\bm{y}^{(k,1)} = \big(  \mathcal{M}(X^{(k)}_1, \theta^{(k)}_1), \ldots ,  \mathcal{M}(X^{(k)}_1, \theta^{(k)}_{\tau})\big),$$ $$\bm{y}^{(k,2)} = \big(  \mathcal{M}(X^{(k)}_2, \theta^{(k)}_1), \ldots ,  \mathcal{M}(X^{(k)}_2, \theta^{(k)}_{\tau})\big). $$ 
\STATE Compute minimal-permutation distance $\psi^{(k)} = \mathsf{d}_{\pi}(\bm{y}^{(k,1)}, \bm{y}^{(k,2)}).$
\ENDFOR 
\STATE Calculate empirical mean $\bar{\psi} = \sum_{k=1}^m \psi^{(k)}_{\tau}/m$. 
\STATE \textbf{Output}: Gaussian covariance $\Sigma_{\bm{B}} = \frac{\bar{\psi}+c}{2v}\cdot \bm{I}_d$.  
\end{algorithmic}
\label{alg: random_MI}
\end{algorithm}

\begin{Thm}
Assume that $\mathcal{M}(X, \theta) \in \mathbb{R}^d$ and $
\|\mathcal{M}(X, \theta)\|_2 \leq r$ for some constant $r$,  and we apply Algorithm \ref{alg: random_MI} to obtain the Gaussian noise covariance $\Sigma_{\bm{B}}$ for a specified mutual information bound $v$. Then, if $\frac{|\bm{\Theta}|}{\tau}$ is an integer, where $|\bm{\Theta}|$ is the total number of the randomness seeds of $\mathcal{M}(\cdot, \theta)$ and $\tau$ is the number of seeds selected in Algorithm \ref{alg: random_MI}, and $c$ is a safe parameter in Algorithm \ref{alg: random_MI}, one can ensure that $\mathsf{MI}(X;\mathcal{M}(X, \theta)+\bm{B}) \leq v$ with confidence at least $1-\gamma$ once $$ m \geq \frac{8r^4\log(1/\gamma)}{c^2}.$$
\label{thm: random_alg}
\end{Thm}

\begin{proof}
See Appendix \ref{app:thm:random_alg}. 
\end{proof}

In Algorithm \ref{alg: random_MI}, our key idea is to do sampling on the randomness seeds across $\bm{\Theta}$ rather than accessing all possible selections of $\theta$. Similar to Algorithm \ref{alg: deter_MI}, we set a safe parameter $c$, which could be arbitrarily small, to lower bound the magnitude of noise added. Given $c>0$, combining the assumption that $\|\mathcal{M}(X,\theta)\|\leq r$, we have a very rough but uniform upper bound on $\mathsf{MI}(X;\mathcal{M}(X,\theta))$ for any possible distributions of $X$ and $\mathcal{M}(X,\theta)$. With this fact, in each iteration in Algorithm \ref{alg: random_MI}, we independently sample two data inputs $\bar{X}^{k} = \{ X^{(k)}_1, X^{(k)}_2\}$, and ${\mathcal{C}}^{(k)} = \{ \theta^{(k)}_1, \ldots ,\theta^{(k)}_\tau\}$ as a random $\tau$-subset selection of $\theta$. It can be proved that 
\begin{equation}
    \label{sampling_MI}
    \mathsf{MI}(X;\mathcal{M}(X,\theta)+\bm{B}) \leq \mathbb{E}_{\bar{X}^{(k)}, {\mathcal{C}}^{(k)}, \theta \in \mathcal{C}^{(k)}} \mathcal{D}_{KL}\big( \mathsf{P}_{\mathcal{M}(X^{(k)}_1, \theta)+\bm{B}} \| \mathsf{P}_{\mathcal{M}(X^{(k)}_2, \theta)+\bm{B}} \big),
\end{equation}
Thus, it suffices to consider a conditional local distribution when $\theta$ is restricted to the $\tau$-subset selected. $\tau$ does not influence the lower bound of sampling complexity $m$ required. However, a larger $\tau$, which also requires higher computational complexity to evaluate more $\mathcal{M}(X, \theta)$, will produce a tighter upper bound. The reason behind this is that the conditional sub-sampled distribution with $\theta \in {\mathcal{C}}^{(k)}$ gets closer to the global distribution with $\theta \in \bm{\Theta}$ as $\tau$ increases. We assume $|\bm{\Theta}|/\tau$ is an integer, i.e., we can split $\bm{\Theta}$ into multiple $\tau$-subsets, to enable the application of minimal-perturbation distance $\mathsf{d}_{\pi}$ (Lemma \ref{lem:matrix_subsampling} in Appendix \ref{app:thm:random_alg} in the full version) 
to give a tighter upper bound of the objective KL-divergence to estimate. 

However, for any proper selection of $\tau$, the upper bound (\ref{sampling_MI}) always holds. Thus, once we have a high-confidence estimation of the expectation on the right side of (\ref{sampling_MI}), we also obtain a high-confidence upper bound on the objective $\mathsf{MI}(X;\mathcal{M}(X,\theta)+\bm{B})$. Further, recall that KL-divergence between  $\mathsf{P}_{\mathcal{M}(X^{(k)}_1, \theta)+\bm{B}}$ and $\mathsf{P}_{\mathcal{M}(X^{(k)}_2, \theta)+\bm{B}}$  for $\theta \in   \mathcal{C}^{(k)}$ is always bounded given $c>0$ and bounded $r$. Therefore, through i.i.d. sampling, we can simply use the empirical mean to approximate the expectation with a high probability bound.

{Before the end of this section, we have some comments on the computation complexity of Algorithms \ref{alg: deter_MI} and \ref{alg: random_MI}. As mentioned before, the deterministic $\mathcal{M}$ is a special case of the randomized processing and Algorithm \ref{alg: random_MI} also applies to the deterministic $\mathcal{M}$ where we may view $\bm{\Theta}$ as only containing a single randomness seed. As a comparison, Algorithm \ref{alg: deter_MI} requires $O(\frac{d+\log(1/\gamma)}{c^2}) $ simulation trials, while Algorithm \ref{alg: random_MI} only requires $O(\frac{\log(1/\gamma)}{c^2})$ simulation trials. Here, we assume the parameters $r$ and $\beta$ are all constants. The reason behind the dependence on the dimensionality $d$ in Algorithm \ref{alg: deter_MI} is because to characterize the optimal Gaussian noise, we need to estimate the power of output distribution $\mathcal{M}(X)$ in each direction in $\mathbb{R}^d$, captured by the spectrum of $\Sigma_{\mathcal{M}(X)}$. Though the estimation on $l_2$-norm based (minimal-permutation) distance is more efficient, as a tradeoff, the noise upper bound produced by Algorithm \ref{alg: random_MI} is looser, $\Theta(\sqrt{d})$, with a strict dependence on $d$. We will give a concrete example for this gap between the noise bound derived from Algorithm \ref{alg: deter_MI} and Algorithm \ref{alg: random_MI} in Example \ref{ex: mean}. We leave a more generic tradeoff between the simulation efficiency and noise control to future work.}

\subsection{Towards Optimal Perturbation}
\label{sec:optimal} 
The algorithms presented so far all produce conservative bounds of perturbation. Indeed, even for mean estimation, we will show (Example \ref{ex: gap} in Section \ref{sec:worst_vs_average}) that the gap between the proposed and the optimal perturbation could be asymptotically large. Indeed, in both Algorithms \ref{alg: deter_MI} and \ref{alg: random_MI}, the proposal of perturbation distribution $\mathsf{P}_{\bm{B}}$ of noise $\bm{B}$ is independent of the selections of input selection $X$ and randomness seed $\theta$, while the optimal perturbation could be dependent on both. Theoretically, we can define the optimal perturbation problem as follows.

\begin{Def}[Optimal Perturbation] Given a data generation distribution $\mathsf{D}$, a mechanism $\mathcal{M}(\cdot, \theta)$, an objective mutual information bound $v$, and a perturbation scheme $\bm{B}$ in a form that $\bm{B}(X, \theta) \sim Q^*(X,\theta)$, i.e., $\bm{B}$ is generated from distribution $Q^*(X,\theta)$ given selections of $X$ and $\theta$, we call the noise distribution $Q^*(X,\theta)$ optimal w.r.t. $(\mathsf{D}, \mathcal{M},\mathcal{K})$ for some loss function $\mathcal{K}$, if 
\begin{equation}
    \begin{aligned}
      Q^*(X,\theta) & = \arg \inf_{Q} \mathbb{E}_{X, \theta, \bm{B} \sim Q} \mathcal{K}\big(\bm{B}(X, \theta)\big), \\
      s.t. \quad \mathsf{MI} & \big(X; \mathcal{M}(X, \theta)+\bm{B}(X, \theta) \big) \leq v. 
    \end{aligned}
    \label{optimal_loss}
\end{equation}
\end{Def}

Optimal perturbation characterizes the least noise (measured by loss function $\mathcal{K}$) that we have to add to achieve the required mutual information bound $v$. For example, in many applications, we are mainly concerned with the expected $l_2$ norm of perturbation, where $\mathcal{K}(\mathcal{B}(X, \theta)) = \|\mathcal{B}(X, \theta)\|_2.$ Though a generic efficient solution to find the optimal perturbation may not exist, we present a  {\em propose-then-verify} framework to approximate the optimal solution. Without loss of generality, we consider the generic setup where $\mathcal{M}(\cdot, \theta)$, $\theta \in \bm{\Theta}$, could be randomized and we consider an equivalent discrete data generation procedure to randomly select $X$ from a set $\mathcal{X}^* = \{X_1, X_2, \ldots ,X_{N}\}$ of size $N$, where the $X_i$s are not necessarily all distinct.

\begin{algorithm}[t]
\caption{$(1-\gamma)$-Confidence Verification of Perturbation Proposal}
\begin{algorithmic}[1]
\STATE \textbf{Input:} A randomized mechanism $\mathcal{M}(\cdot,\theta): \mathcal{X}^* \to \mathbb{B}^d_{r}$, $\theta \in \bm{\Theta}$, $\mathcal{X}^* = \{X_1, X_2, \ldots ,X_{N}\}$, a perturbation proposal $\bm{B}(X, \theta) \sim Q(X, \theta)$, sampling complexity $m$, security parameters $c$, $\tau_1, \tau_2$ and $\tau_3$.    
\FOR{$k=1,2,\ldots,m$}
\STATE Randomly sample $\tau_1$ selections of $X$ from $\mathcal{X}^*$, denoted by  $\bar{X}^{(k,1)}=\{X^{(k,1)}_1, \ldots ,X^{(k,1)}_{\tau_1}\},$ and independently sample another $\tau_2$ selections of $X$, denoted by  $\bar{X}^{(k,2)}=\{X^{(k,2)}_1, \ldots ,X^{(k,2)}_{\tau_2}\}. $ 
\STATE Randomly sample $\tau_3$ selections of $\theta$ from $\bm{\Theta}$, denoted by $\mathcal{C}^{(k)} = \{\theta^{(k)}_1, \theta^{(k)}_2, \ldots , \theta^{(k)}_{\tau_3}\}.$
\STATE Define two sets of distributions for independent $\bm{B}_c \sim \mathcal{N}(0, c \cdot \bm{I}_d)$, 
$$\mathsf{P}^{(k,1)}_i = \sum_{j=1}^{\tau_3} \frac{1}{\tau_3} \cdot \big( \mathsf{P}_{\mathcal{M}(X^{(k,1)}_{i},  \theta^{(k)}_j)+ \bm{B}(X^{(k,1)}_{i}, \theta^{(k)}_j)+ \bm{B}_c}\big), i=1,2,\ldots,\tau_1,$$
$$\mathsf{P}^{(k,2)}_l = \sum_{j=1}^{\tau_3} \frac{1}{\tau_3} \cdot \big( \mathsf{P}_{\mathcal{M}(X^{(k,2)}_{l}, \theta^{(k)}_j) + \bm{B}(X^{(k,2)}_{l}, \theta^{(k)}_j)+\bm{B}_c} \big),   l=1,2,\ldots,\tau_2.$$ 
\STATE Compute the average of KL-divergences,
\begin{equation}
    \label{emperical_MI}
    \psi^{(k)} = \sum_{i=1}^{\tau_1} \big( \frac{1}{\tau_1} \cdot  \mathcal{D}_{KL}(\mathsf{P}^{(k,1)}_i\|\sum_{l=1}^{\tau_2} \frac{1}{\tau_2} \cdot \mathsf{P}^{(k,2)}_l ) \big).
\end{equation}
\ENDFOR 
\STATE Calculate empirical mean $\bar{\psi} = \sum_{k=1}^m \psi^{(k)}/m$. 
\STATE \textbf{Output}:  $\bar{\psi}$. 
\end{algorithmic}
\label{alg: verification}
\end{algorithm}

\begin{Thm}
Assume that $\mathcal{M}(X, \theta) \in \mathbb{R}^d$ and $
\|\mathcal{M}(X, \theta)\|_2 \leq r$ for some constant $r$, and for any perturbation proposal $Q(X, \theta)$, $\bm{B}(X, \theta) \sim Q(X, \theta)$, we apply Algorithm \ref{alg: verification} to obtain $\bar{\psi}$. Then, if $\frac{N}{\tau_2}$ and $\frac{|\bm{\Theta}|}{\tau_3}$ are integers, one can ensure that $\mathsf{MI}\big(X;\mathcal{M}(X, \theta)+\bm{B}(X, \theta) + \bm{B}_c) \leq \bar{\psi}+\beta$ for an independent Gaussian noise of the form $\bm{B}_c \sim \mathcal{N}(0, c\cdot \bm{I}_d)$ with confidence at least $1-\gamma$ when $$ m \geq \frac{2\log(1/\gamma)}{\beta^2}\cdot \big((2r^2/c)^2/\tau_1 + \beta/3 \cdot (2r^2/c) \big).$$
\label{thm: verification}
\end{Thm}
\begin{proof}
See Appendix \ref{app:thm: verification}. 
\end{proof}

In Algorithm \ref{alg: verification}, we can give a high-confidence claim on the security parameters produced by any perturbation proposal $Q(X, \theta)$. In practice, when input set $X$ is large and representative enough, and the mechanism $\mathcal{M}$ can stably learn populational information from $X$, the distributions of $\mathcal{M}(X)$ would be close to each other for most selections of $X$. Therefore, an empirical approximation framework of optimal perturbation can be described as follows. 
\begin{enumerate}
    \item We follow Steps 3-5 in Algorithm \ref{alg: random_MI} to sample and evaluate on some large enough sets with respect to the selections of $X$ and $\theta$, respectively, and optimize the noise distribution to minimize its $\mathcal{K}$ loss  in (\ref{optimal_loss}) such that the estimated divergence defined in (\ref{emperical_MI}) is smaller than $(1/O(1)) \cdot v$. Here, $v$ is the desired mutual information bound. 
    \item Provided the locally-optimized noise proposal $\hat{Q}$, we can apply Algorithm \ref{alg: verification} to test $\hat{Q}$. If $\bm{B} \sim \hat{Q}$ does not produce the desired mutual information bound $v$, we adjust $\bm{B} \sim \hat{Q}$ in a form $\bm{B}+ \bm{B}_{\alpha}$, where $\bm{B}_{\alpha}$ is an independent Gaussian noise of the form $\mathcal{N}(0, \alpha \cdot \bm{I}_d)$. We then perform binary search on $\alpha$ via Algorithm \ref{alg: verification} until we find a proper value. 
\end{enumerate}

\section{Composition}

\label{sec: composition}
In this section, we study the composition of PAC Privacy. Composition is an important requirement of a practical security definition. Since the privacy loss of any release cannot be revoked, it is necessary to keep track of the accumulated leakage. Ideally, the respective privacy loss of multiple releases is expected to be analytically aggregated to produce an upper bound on the accumulated loss incurred. For example, it is well known that for two mechanisms $\mathcal{M}_1$ and $\mathcal{M}_2$, which each satisfy $(\epsilon_0, \delta_0)$-DP and use {\em independent} randomness, the composite mechanism $\bar{\mathcal{M}}(X) = \big(\mathcal{M}_1(X), \mathcal{M}_2(X)\big)$ then satisfies $(2\epsilon_0, 2\delta_0)$-DP. With advanced composition, one may obtain a sharpened bound $\tilde{O}(\sqrt{T\log \delta'}\epsilon_0, T\delta_0+\delta')$-DP for a free parameter $\delta'>0$ after a $T$-fold composition of mechanisms each satisfying $(\epsilon_0,\delta_0)$-DP. 

Composition also plays an important role in a classic privacy regime, especially for DP, to analyze complicated algorithms. As mentioned before, since the worst-case sensitivity is hard to compute, in general, to enable non-trivial privacy analysis, a complex algorithm usually needs to be artificially decomposed into multiple, relatively simpler operations. Composition then provides a way to upper bound the entire leakage by aggregating per-operation privacy loss. DP-SGD, one of the most widely-applied DP optimization methods, is an example. The DP analysis is developed in a compositional fashion where it is pessimistically assumed that intermediate gradients across each iteration are released. The entire iterative optimization is thus decomposed into multiple per-sample gradient aggregations \cite{DP-deeplearning}. Accordingly, the DP composition result also shows a simple way to determine the scale of noise required, which increases as $\tilde{O}(\sqrt{T})$ to produce a fixed privacy guarantee. 

However, we have to stress that in DP, the above described results only hold when the randomness in different mechanisms is independent. This limitation makes tight Differential Privacy in many applications intractable. For example, if $\mathcal{M}_{[1:T]}$ represent $T$ sequential processings on subsampled data from the sensitive $X$ set but the subsampling schemes in different $\mathcal{M}_{[1:T]}$ are correlated, then more involved compositional analysis is needed.  In comparison, our PAC Privacy analysis framework can indeed handle composition of arbitrary (possibly correlated) mechanisms. Besides, PAC Privacy allows us to measure the end-to-end privacy leakage of the output from an arbitrary (black-box) $\mathcal{M}$. Therefore, there is no need to artificially decompose an algorithm and control leakage step-by-step. This is one key difference compared to DP: we use composition in PAC Privacy only when multiple releases are necessary.

We first consider the simplest case for $T$ mechanisms $\mathcal{M}_{[1:T]}$ where the data generation and randomness are both independent. By the chain rule of mutual information, we then have the following equation, 
\begin{equation}
    \label{MI_sum}
    \mathsf{MI}(X_{[1:T]}; \mathcal{M}_1(X_1, \theta_1), \ldots ,\mathcal{M}_T(X_T, \theta_T)) = \sum_{t=1}^T \mathsf{MI}(X_t; \mathcal{M}_t(X_t, \theta_t)),
    \end{equation}
which states that mutual information can be simply summed. Therefore, by plugging $\sum_{k=1}^T \mathsf{MI}(X_k; \mathcal{M}_k(X_k))$ into (\ref{generic_fano}) in Theorem \ref{thm: PAC_joint}, we can obtain an upper bound on the posterior advantage $\Delta_{KL} \delta$, and accordingly a lower bound of $\delta$. 

Now, we turn to consider the more generic scenario where $\mathcal{M}_{[1:T]}$ apply independent randomnesses but share the same input $X$ generated from $\mathsf{D}$. Unfortunately, (\ref{MI_sum}) does not hold anymore in this setup and mutual information does {\em not} enjoy a triangle inequality: $\mathsf{MI}(X; \mathcal{M}_1(X), \mathcal{M}_2(X))$  is not upper bounded by $\mathsf{MI}(X; \mathcal{M}_1(X)) + \mathsf{MI}(X; \mathcal{M}_2(X))$ in general. This is also true for most other selections of $f$-divergence and thus PAC Privacy, in general, cannot be simply summed up to produce an upper bound. Fortunately, we will show below that the upper bound on mutual information studied in Algorithm \ref{alg: random_MI} and Theorem \ref{thm: random_alg} can be simply summed up to form a new bound under the composition. In other words, the PAC Privacy of privatized algorithms using the proposed automatic analysis framework {\em does} enjoy a simple summable composition bound. 

\begin{Thm}
\label{thm: composition}
    For arbitrary $T$ mechanisms $\mathcal{M}_{t}(X, \theta_t)$, $t=1,2,\ldots,T$, whose randomness seeds $\theta_t$ are independently selected, and data generation $X \sim \mathcal{D}$, if via Algorithm \ref{alg: random_MI}, 
    for each $\mathcal{M}_{t}$ a noise $\bm{B}_t$ is determined and independently generated such that with confidence $(1-\gamma_t)$, $\mathsf{MI}(X; \mathcal{M}_t(X, \theta_t)) \leq v_t$, then with confidence $(1-\sum_{t=1}^T \gamma_t)$, 
    $\mathsf{MI}(X; \mathcal{M}_1(X, \theta_1), \ldots ,\mathcal{M}_T(X, \theta_T)) \leq  \sum_{t=1}^T v_t. $
\end{Thm}

\begin{proof}
    See Appendix \ref{app: thm: composition}.
\end{proof}
From Theorem \ref{thm: composition}, when the mechanisms $M_{[1:T]}$ are identical, we need to increase the noise by a factor at most $\sqrt{T}$ to ensure the same bound of the $T$-composite mutual information compared to the case of a single mechanism. Before proceeding, we have several important remarks. Even if one does not know the mechanisms to be composed in advance, one can choose $v_t$ for each $\mathcal{M}_t$ in an online fashion and use Algorithm \ref{alg: random_MI} to determine appropriate noise for each exposure. Though the composition bound in Theorem \ref{thm: composition} is simple, tighter analysis exists using the proposed automatic framework. It is noted that one can always build a new joint mechanism $\bar{\mathcal{M}}$ which outputs $\big(\mathcal{M}_{1}(X, \theta_1), \ldots ,\mathcal{M}_{T}(X, \theta_T)  \big)$, and simply apply Algorithm \ref{alg: random_MI} on $\bar{\mathcal{M}}$ (or Algorithm \ref{alg: deter_MI} in the deterministic case). 

In the following, we move on to the even more generic case for adaptive composition with correlated randomness. We consider $T$ mechanisms $\mathcal{M}_{[1:T]}$ which are implemented in a sequential manner with possibly correlated selection of randomness seeds. More formally, the output $\mathcal{M}_{t}(X, {\theta}_t, \mathcal{M}_{j}(X, \theta_{j}), j=1,2, \ldots,t-1)$ is defined by $X$, the randomness ${\theta}_t$, and the previous outputs $\mathcal{M}_{j}(X, \theta_{j})$ for $j=1,2, \ldots,t-1$. Without loss of generality, in the following, we use $\mathcal{M}_t(X, \bar{\theta}_t)$ to denote the $t$-th mechanism where $\bar{\theta}_t = (\theta_1, \ldots ,\theta_t)$ represents the joint randomness seed and $\theta_t$ may be selected dependent on $\bar{\theta}_{t-1}$. On the other hand, we are given a sequence of privacy budgets $v_0=0 < v_1 < v_2 < \ldots < v_T$ and our goal is to ensure that for any $t = 1,2, \ldots,T$, the accumulated privacy loss measured as mutual information is upper bounded by $v_t$ with high confidence.

To this end, we propose an online implementation of Algorithm \ref{alg: random_MI} described as Algorithm \ref{alg: sequential_random_MI}. The key idea is still to incrementally aggregate existing mechanisms to form a joint one, that we can apply similar techniques in Algorithm \ref{alg: random_MI} to analyze. Algorithm \ref{alg: sequential_random_MI} has an initialization step (lines 2-6) where $m$ input pairs $(X^{(k)}_1, X^{(k)}_2), k=1,2,...,m$, are randomly generated from $\mathsf{D}$, and for each pair we select $\tau$ randomness seeds $\bar{\theta}^{(k)}_{1,l}$, $l=1,2,\ldots,\tau$, for $\mathcal{M}_1$ (lines 3-4). In the online phase, for steps $t > 1$, we generate new randomness seeds $\bar{\theta}^{(k)}_{t,l}$ conditional on $\bar{\theta}^{(k)}_{t-1,l}$ determined previously (line 9) and continue the simulations on newly-incoming $\mathcal{M}_{t}$ over the same $m$ input pairs selected (line 11). This saves computation, avoiding reevaluating $\mathcal{M}_{[1:t-1]}$ on fresh samples and randomness. Given previous noise schemes $\bm{B}_{[1:(t-1)]}$ determined by Algorithm \ref{alg: sequential_random_MI} such that 
$\mathsf{MI}(X; {\mathcal{M}}_{1}(X, \bar{\theta}_{1}) + \bm{B}_{1}, \ldots , {\mathcal{M}}_{t-1}(X, \bar{\theta}_{t-1}) + \bm{B}_{t-1}) \leq v_{t-1}$, we can always determine some independent Gaussian noise $\bm{B}_{t}$ such that 
$$\mathsf{MI}(X; {\mathcal{M}}_{1}(X, \bar{\theta}_{1}) + \bm{B}_{1}, \ldots , {\mathcal{M}}_{t-1}(X, \bar{\theta}_{t-1}) + \bm{B}_{t-1},{\mathcal{M}}_{t}(X, \bar{\theta}_{t}) + \bm{B}_{t} ) \leq v_{t},$$
for any $v_{t} > v_{t-1}$ with high confidence. The following theorem describes the tradeoff between the confidence level and sampling complexity.

\begin{Thm}
   Assume that $\mathcal{M}_t(X, \bar{\theta}_t)\in \mathcal{R}^d$ and $\|\mathcal{M}_t(X, \bar{\theta}_t) \|_2 \leq r$ uniformly for some constant $r$, for $t=1,2,\ldots,T$, and we apply Algorithm \ref{alg: sequential_random_MI} which returns a sequence of Gaussian noise covariances $\Sigma_{\bm{B}_t}$ in an online setup. Then, one can ensure that $\mathsf{MI}(X; {\mathcal{M}}_{1}(X, \bar{\theta}_{1}) + \bm{B}_{1}, \ldots ,{\mathcal{M}}_{t}(X, \bar{\theta}_{t}) + \bm{B}_{t} ) \leq v_{t}$ for any $t=1,2,\ldots,T,$ with confidence at least $1-\gamma$ once $ m \geq \frac{8r^4\log(T/\gamma)}{c^2}$. $c$ as before is a security parameter.
\label{thm: online_random_alg} 
\end{Thm}
\begin{proof}
    See Appendix \ref{app: thm: online_random_alg}.
\end{proof}
As a final remark, though the deterministic algorithm is a special case of the randomized algorithm (with a single randomness seed as discussed earlier) and we can always use Algorithm \ref{alg: random_MI} for privacy analysis, Algorithm \ref{alg: deter_MI} can usually produce 
tighter noise bounds. However, the composition of the PAC Privacy result derived from Algorithm \ref{alg: deter_MI} is more involved, which depends on the eigenspace estimated. We leave the generalization of Theorem \ref{thm: composition} for Algorithm \ref{alg: deter_MI}   as an open problem. However, with simulations on incrementally aggregated mechanisms to be composed, one can similarly build an online version of Algorithm \ref{alg: deter_MI} to determine the perturbation sequentially given increasing privacy budget for composition.

\begin{Ex}[Composition] We take the {\em Hospital} dataset from Mathworks Sample \footnote{https://www.mathworks.com/help/stats/sample-data-sets.html} as an example. This set contains the records of 100 patients and we consider two mechanisms $\mathcal{M}_1$ and $\mathcal{M}_2$ which return the average of ages, weight, and the average of blood pressure range, respectively. We consider randomly selecting 50 individual samples to form the input $X$ and via Algorithm \ref{alg: random_MI} we need Gaussian noises $\mathbb{E}[\|\bm{B}_1\|_2]= 6.85$ and $\mathbb{E}[\|\bm{B}_2\|_2]= 2.60$ to ensure that $\mathsf{MI}(X; \mathcal{M}_1(X)+\bm{B}_1)\leq 0.5$ and $\mathsf{MI}(X; \mathcal{M}_2(X)+\bm{B}_2)\leq 0.5$. With the simple composition bound, we have that 
$\mathsf{MI}(X; \mathcal{M}_1(X)+\bm{B}_1,\mathcal{M}_2(X)+\bm{B}_2) \leq 1$. As a comparison, by applying the automatic analysis on the joint mechanism $\bar{M}=(\mathcal{M}_1, \mathcal{M}_2)$, we obtain a tighter noise bound where there exists certain Gaussian noise $\bar{\bm{B}}$ with $\mathbb{E}[\|\bar{\bm{B}}\|_2]= 4.57$ to ensure $\mathsf{MI}(X; \bar{\mathcal{M}}(X)+\bar{\bm{B}})\leq 1$. Using the simple composition, it is noted that the noise magnitude $\mathbb{E}[\|(\bm{B}_1, \bm{B}_2)\|_2] = 7.32,$ when we handle $\mathcal{M}_1$ and $\mathcal{M}_2$ separately with lower-complexity privacy analysis.  
\label{ex: composition}
\end{Ex}

\begin{algorithm}[t]
\caption{Online Noise Determination of Adaptively-Composed Randomized Mechanisms $\mathcal{M}_{[1:T]}$}
\begin{algorithmic}[1]
\STATE \textbf{Input:} A set of randomized mechanisms $\mathcal{M}_t(\cdot, \bar{\theta}_t): \mathcal{X}^{*} \times \bm{\Theta}^* \to \mathbb{B}^d_{r}$, $\bar{\theta}_t \in \bm{\Theta}_1 \times \ldots \times \bm{\Theta}_t$, where $\mathcal{M}_t(\cdot, \bar{\theta}_t)$ is provided at time slot $t$ for $t=1,2,\ldots,T$, data distribution $\mathsf{D}$, sampling complexity $m$, security parameters $c$, $\tau$, and privacy budget $\{v_0=0, v_1, v_2, \ldots , v_T\}$ such that $0< v_1 < v_2 < \ldots < v_T$.
\FOR{$k=1,2,\ldots,m$}
\STATE From distribution $\mathcal{D}$, independently sample ${X}^{(k)}_1$ and
${X}^{(k)}_2.$
\STATE Initially randomly select $\tau$ seeds, denoted by $\bar{\mathcal{C}}^{(k)}_{1} = \{\bar{\theta}^{(k)}_{1,1}, \bar{\theta}^{(k)}_{1,2}, \ldots , \bar{\theta}^{(k)}_{1,\tau}\}$, $\bar{\theta}^{(k)}_{1,l} \in \bm{\Theta}_1$.
\ENDFOR 

\FOR{$t=1,2,\ldots,T$}
\FOR{$k=1,2,\ldots,m$}
\IF{$t > 1$}
\STATE Conditional on $\bar{\theta}^{(k)}_{t-1,l}$, randomly select $\theta^{(k)}_{t,l} \in \bm{\Theta}_{t}$ and form a new joint seed $\bar{\theta}^{(k)}_{t,l} = (\bar{\theta}^{(k)}_{t-1,l}, \theta^{(k)}_{t,l} )$, for $l=1,2,\ldots,\tau$. 
\ENDIF
 
\STATE Compute and record $$\bm{y}^{(k,1)}_t = \big(  \mathcal{M}_{t}(X^{(k)}_1, \bar{\theta}^{(k)}_{t,1}), \ldots ,  \mathcal{M}_{t}(X^{(k)}_1, \bar{\theta}^{(k)}_{t,\tau}) \big),$$ 

$$\bm{y}^{(k,2)}_t = \big(  \mathcal{M}_{t}(X^{(k)}_2, \bar{\theta}^{(k)}_{t,1}), \ldots ,  \mathcal{M}_{t}(X^{(k)}_2, \bar{\theta}^{(k)}_{t,\tau}) \big). $$ 
\STATE Compute averaged $l_2$-norm pairwise distance $\psi^{(k)}_{t} = {\|y^{(k,1)}_t - y^{(k,2)}_t \|^2_2}/{\tau}.$ 
\ENDFOR  
\STATE Calculate empirical mean $\bar{\psi}_{t} = \sum_{k=1}^m \psi^{(k)}_{t}/m$. 
\STATE \textbf{Output}: Gaussian covariance $\Sigma_{\bm{B}_t} = \frac{\bar{\psi}_{t}+c}{2(v_t-v_{t-1})}\cdot \bm{I}_d$.  
\ENDFOR
\end{algorithmic}
\label{alg: sequential_random_MI}
\end{algorithm}

\section{Comparison between Adversarial and Instance-based Worst Case}
\label{sec:worst_vs_average}
We start by a lower bound of the Gaussian noise if we want to ensure an input-independent upper bound of the mutual information.  
\begin{Thm}[Lower Bound of Adversarial Gaussian Mechanism]
Given that $\|\mathcal{M}(X)\|_2 \leq r$, for an arbitrary Gaussian perturbation mechanism $Q(y) = \mathcal{N}(0,\Sigma(y))$, i.e., when $\mathcal{M}(X)=y \in \mathbb{B}^d_{r}$, we generate $\bm{B} \sim Q(y)$ and publish $\mathcal{M}(X)+\bm{B}$. Given $v=o(1)$, if
$$ \sup_{\mathcal{M}, \mathsf{P}_X} \mathsf{MI}(X; \mathcal{M}(X)+\bm{B}) \leq v, $$
then for any $y \in \mathbb{B}^d_{r}$ and $\bm{B} \sim Q(y)$, $\mathbb{E}[\| \bm{B}\|_2] = \Omega(r\sqrt{d}/\sqrt{v}).$
\label{thm: lower-bound}
\end{Thm}

\begin{proof}
See Appendix \ref{app: thm: lower-bound}.
\end{proof}

From Theorem \ref{thm: lower-bound}, we see that if we want a data-independent Gaussian perturbation to ensure bounded mutual information in the adversarial worst case, we must add a noise of magnitude of $\Omega(r\sqrt{d}/\sqrt{v})$  to {\em any possible output} $\mathcal{M}(X)$\footnote{In Theorem \ref{thm: lower-bound}, we restrict the noise distribution to be Gaussian. We leave a generic lower bound for arbitrary noise distribution as an open problem.}. On the other hand, from Theorem \ref{thm:deter_m}, we know this bound is also tight where a Gaussian noise of magnitude $O(r\sqrt{d}/\sqrt{v})$ can universally ensure the mutual information is bounded by $v$. This matches our intuition that without any other assumptions if we want to hide arbitrary information disclosure (bounded by $r$ in $l_2$ norm), we need noise on each direction in $\mathbb{R}^d$ in a scale of $\Omega(r)$, which finally produces the $\Omega(r\sqrt{d})$ magnitude. This, in general, makes the perturbed release useless. Therefore, for a particular processing mechanism $\mathcal{M}$, when the practical data is far away from the adversarial worst case or the data-independent worst-case proof cannot be tightly produced, there could be a huge utility compromise relative to the PAC Privacy measurement from an instance-based angle. 

\begin{Ex}[Membership Inference from Mean Estimation] We suppose $\mathsf{U}$ to be CIFAR10 $\cite{cifar10}$, a canonical test dataset commonly used in computer vision, which consists of $60,000$ $32\times 32$ color images. We normalize each pixel to within $[0,1]$ and take each sample as a 3072-dimensional vector. We suppose each of the $60,000$  $3072$-dimensional samples is i.i.d. randomly selected into the set $X$, where the expected cardinality of $|X|$ equals $30,000$, to conduct the mean estimation analysis, where $\mathcal{M}(X)$ simply returns the sum of $X$ divided by $30,000$, as an unbiased mean estimator. Via Algorithm \ref{alg: deter_MI}, we have that an independent Gaussian noise whose $\mathbb{E}[\|\bm{B}\|_2]= 0.28$ is sufficient to ensure $\mathsf{MI}(X; \mathcal{M}(X)+\bm{B})\leq 1$. On the other hand, in the adversarial worst-case setup, the data of normalized pixel suggests that the $l_{\infty}$-norm of the sensitivity is bounded by $1$. The mutual information bound requires $\xi=1/60000$-zCDP \cite{bun2016concentrated}, in turn requiring a noise $\mathbb{E}[\|\bm{B}\|_2] = 17.7$, which is 63$\times$ larger than the instance-based analysis. As a comparison, if we apply Algorithm \ref{alg: random_MI} to compute the noise bound, it is required that $\mathbb{E}[\|\bm{B}\|_2]=3.94$ which is looser compared to $0.28$ obtained from Algorithm \ref{alg: deter_MI}, though it is more efficient and only requires the estimation on the (pair-wise) $l_2$ distance among simulation outputs. 
\label{ex: mean}
\end{Ex}
\noindent In Example \ref{ex: mean}, we already utilize the property that each data point is independently sampled and the individual sensitivity of the averaging operator has a closed form $O(1/n)$ to get a tighter bound for the adversarial worst case \cite{bun2016concentrated}\footnote{If each entry of $X$ is independently generated, and $\mathcal{M}$ satisfies $\xi$-Concentrated Differential Privacy (CDP), then $\mathsf{MI}(X; \mathcal{M}(X)) \leq n\xi$ \cite{bun2016concentrated}.  }. However, in general, when the sampling is not independent or the sensitivity is intractable, one may need to adopt the loose noise bound shown in Theorem \ref{thm: lower-bound}, where a noise in a scale of $2.2\times 10^3$ is needed. Below, we provide such an example where producing a non-trivial adversarial worst-case proof remains open. 

\begin{Ex}[Private Machine Learning] We consider measuring data leakage of deep learning on practical data. In general, in machine learning, our goal is to  optimize the weights/parameters $w$ of some model, viewed as some function $G(\cdot, w)$, to minimize the loss $\min_{w} \mathbb{E}_{(x,y)}\mathcal{L}\big(G(x,w), y \big)$. Here, $\mathcal{L}$ represents some loss function, while $x$ and $y$ represent the feature and label of a sample, respectively. Ideally, we expect the trained out model $G(x,w)$ to predict the true label $y$ well.

In this example, we consider training a three-layer fully-connected neural network on the MNIST dataset $\cite{MNIST}$. MNIST contains $70,000$  $28\times 28$ handwritten-digit images. We consider a data generation $X$ by randomly sampling $35,000$ samples out of the entire data set. The neural network has three layers, where the weights of the first layer are a $784\times 30$ matrix, those of the second layer are a $30 \times 30$ matrix, and those of the last layer are a $30 \times 10$ matrix. We select the activation function between layers to be the Relu function and use cross-entropy as the loss function, as a common setup in deep learning \cite{lecun2015deep}. The total number of parameters in this network is $24,790$. Let the mechanism $\mathcal{M}$ correspond to running full-batch gradient descent for $1,500$ iterations with step size $0.05$ and outputting the final weight obtained. From Algorithm \ref{alg: deter_MI}, we have that an independent Gaussian noise $\mathbb{E}[\|\bm{B}\|_2]=3.7$ is sufficient to ensure $\mathsf{MI}(X; \mathcal{M}(X)+\bm{B}) \leq 1$. Non-privately, the trained-out neural network achieves $94.8\%$ classification accuracy; under the perturbation to ensure PAC Privacy, we achieve an accuracy of $93.5\%$, i.e., a slight compromise. Unfortunately, for the adversarial data-independent worst case, the sensitivity/stability of generic non-convex optimization remains open \footnote{An upper bound on individual sensitivity of Empirical Risk Minimization is only known for strongly-convex optimization \cite{chaudhuri2011differentially} with an additional Lipschitz assumption and the loss function needs to be a sum of individual losses on each sample.}. 
\end{Ex}

Before the end of this section, we want to point out that there could be an asymptotically large gap between our proposed conservative perturbation and the optimal noise required, in particular for the individual privacy case. 

\begin{Ex}[Gap to Optimal Perturbation] We return to the problem of mean estimation of Gaussian distributions mentioned in Section \ref{sec: Intro}, where we assume each $x_i \in \mathbb{R}^d \sim \mathcal{N}(\mu, \Sigma)$ and $\mathcal{M}(X) = 1/n \cdot \sum_{i=1}^n x_i$. We focus on the individual privacy $\mathsf{MI}(x_i; \mathcal{M}(X))$ and set out to quantify the least noise needed. For any independent Gaussian noise, $\bm{B}' \sim \mathcal{N}(0, \Sigma_{\bm{B}'})$, we have 
\begin{equation}
    \label{gau_true}
     \mathsf{MI}(x_i; \mathcal{M}(X)+\bm{B}') = \frac{1}{2} \cdot \log\big( det\big( (\frac{1}{n}\cdot \Sigma + \Sigma_{\bm{B}'}) \cdot (\frac{(n-1)}{n^2} \cdot \Sigma + \Sigma_{\bm{B}'})^{-1} \big)\big).
\end{equation}
On the other hand, when we apply Algorithm \ref{alg: deter_MI} and measure $\mathsf{MI}(X; \mathcal{M}(X))$ as a uniform upper bound,   
\begin{equation}
    \label{gau_est}
     \mathsf{MI}(X; \mathcal{M}(X)+\bm{B}) = \frac{1}{2} \cdot \log\big( det\big( (\frac{1}{n}\cdot \Sigma + \Sigma_{\bm{B}}) \cdot \Sigma^{-1}_{\bm{B}} \big)\big).
\end{equation}
From $(\ref{gau_true})$ and $(\ref{gau_est})$, we can determine the corresponding Gaussian noises to ensure a mutual information bound $v$ and their gap shows as a lower bound of the gap between $\bm{B}$ and the optimal perturbation $\bm{B}_o$, since in this example we already restrict the noises to be Gaussian. With some calculation, we have
\begin{equation}
\small{
    \label{noise_gap}
    \mathbb{E}[\|\bm{B}\|^2_2] - \mathbb{E}[\|\bm{B}_{o}\|^2_2] \geq \left \{
    \begin{aligned}
   &\mathbb{E}[\|\bm{B}\|^2_2] = \frac{1}{2v} \cdot  (\sum_{j=1}^d \sqrt{\lambda_j})^2,  &  v \geq \frac{d}{n-1} \\
&\sum_{j=1}^d \min\{ \sqrt{\lambda_j}(\sum_{j=1}^d\sqrt{\lambda_j})/(2v), (n-1)\lambda_j\}, &  v < \frac{d}{n-1}.
\end{aligned}
\right.}
\end{equation}
Here, $\lambda_{[1:d]}$ are the eigenvalues of $\Sigma$. From $(\ref{noise_gap})$, when the privacy budget $v \geq d/(n-1)$, there is even no need to add noise; in the high privacy regime, such a gap could be $\Omega(n\Tr(\Sigma_{\mathcal{M}(X)})).$ Therefore, there is still much room to improve the current analysis.
\label{ex: gap}
\end{Ex}

\section{Related Works}
\label{sec: related-work}

\subsection{Differential Privacy} 
\label{sec: DP}
DP is the most successful and practical information-theoretic privacy notion over the last two decades. Lower bounds both in DP noise \cite{geng2015optimal}, \cite{hardt2010geometry} and Theorem \ref{thm: lower-bound} suggest that, in general, there is no free lunch for input-independent privacy guarantees -- to produce meaningful security parameters, the perturbation could be much larger than the objective disclosure and destroy the released information. This is especially true in the high-dimensional case, known as the curse of dimensionality. DP and our PAC Privacy propose two different solutions to overcome this challenge.  

Differential Privacy (DP) can be most naturally interpreted in a binary hypothesis form. Consider two adjacent datasets $S_0$ and $S_1$ which only differ in one individual datapoint and two hypotheses: $H_0$ where the input set is $S_0$, and $H_1$ where the input set is $S_1$. A strong DP guarantee implies that the hypothesis testing error (a combination of Type I and II errors) is large \cite{dong2019gaussian}. In general, when the adversarial objective is specified to identify the secret from a class of candidates, such inference can be naturally characterized by a (multiple) hypothesis testing model, and one natural privacy definition is to measure the lower bound of the probability when the adversary outputs a wrong guess \cite{MAP}.

The beauty of DP in bridging meaningful privacy and tolerable utility compromise is to focus on an individual datapoint rather than the entire sensitive dataset. DP stays within the data-independent privacy regime. $\epsilon$-DP for constant $\epsilon$ cannot make meaningful guarantees for the entire data, because the security parameter grows linearly with $n\epsilon$ for the privacy concern of $n$ datapoints (group privacy). However, a noise of magnitude $O(\sqrt{d}/n)$ can hide any individual among the $n$-population if one can ensure an $O(1/n)$ sensitivity. Therefore, once the population size $n$ is large enough and we can ensure the contribution from each individual to the output is limited enough, reasonably small perturbation can ensure good individual privacy, interpreted as some constant hypothesis test error \cite{balle2020hypothesis}. However, as mentioned before, the worst-case sensitivity proof becomes the key step in DP, which is generally not easy beyond simple aggregation.

PAC Privacy tackles this challenge by enforcing the privacy measurement to be instance-based to get rid of the restriction from data-independent impossibilities. In terms of the objective to protect, PAC Privacy is stronger and more generic, where now the privacy of {\em the entire dataset} could be preserved once the processing mechanism $\mathcal{M}(X)$ learns or behaves stably based on some population property of the underlying distribution $\mathsf{D}$, from which the sensitive data $X$ is generated. Individual Privacy is only a special case in PAC Privacy. One can imagine that if the output of $\mathcal{M}$ is only dependent on the distribution $\mathsf{D}$, $X$ is then independent of $\mathcal{M}(X)$ and we achieve perfect secrecy with respect to the whole dataset. In particular, for applications in machine learning, PAC Privacy would benefit from strong learning algorithms $\mathcal{M}$ of good generalization. More details can be found in Section \ref{sec: others}. 

However, PAC Privacy also has limitations since it is not applicable when we cannot simulate or we don't have access to the data generation of $X$. In contrast, data-independent privacy (if it can be proven) can always work with meaningful interpretation regardless of any assumptions on priors. But for most statistical data processing, where we are allowed to conduct the analysis on subsampled data, this limitation can be overcome. Even for a given deterministic dataset, through sub-sampling we can enforce that the selected samples enjoy entropy and PAC Privacy provides a strong security interpretation: even the adversary, who has full knowledge of the data pool, cannot infer too much about selected samples in the processing mechanism $\mathcal{M}$. 

\subsection{Other Privacy Metrics and Learning Theory Quantities}
\label{sec: others}
\textbf{Mutual Information (MI) based Attack and Privacy Definition}: MI is a well-studied measure on the dependence between two random variables quantified in information bits. Naturally, one may consider the MI between the secret and the release/leakage for security purposes. A large number of existing works apply MI to construct side-channel attacks \cite{batina2011mutual,MIattack2008} or even simply adopt MI as a privacy definition \cite{chatzikokolakis2010statistical,MIdef2014}. For example, given known plaintexts, \cite{MIattack2008} built a distinguisher based on empirical MI estimation where the adversary proposes a guessing of the encryption key with the highest correlation score regarding the side-channel leakage. \cite{chatzikokolakis2010statistical} showed further asymptotic analysis of empirical MI estimations under discrete randomness. Though we use $f$-divergence, where MI based on KL-divergence is a special case, to develop impossibility results as the foundation of PAC Privacy, our motivations and results are very different from those prior works. On one hand, PAC Privacy is rigorously resistant to arbitrary adversary attacks\footnote{Besides MI, there are many other efficient side-channel attacks based on different statistical tests such as Pearson correlation \cite{brier2004correlation}.}, without any restrictions on the adversary's inference strategy or the computation power. In addition, PAC Privacy does {\em not} use MI or other information metrics, such as Fisher information \cite{fisher2017}, discussed below, or anonymity set \cite{anonymity1,anonymity2}, as the leakage measurement, which, in general, lack semantic security interpretation. Unfortunately, except DP, most of those metrics are not semantic. For example, recent work \cite{GDPR} argued that the concept of anonymity set does not resist singling-out attacks and does not produce desirable individual privacy guarantees. In contrast, the entire theory developed for PAC Privacy is devoted to answering the fundamental question of how to quantify the hardness of an arbitrary adversarial inference task. 

\noindent \textbf{DP under Adversary Uncertainty}: With a different motivation to add less noise, there is a line of works such as distributional DP \cite{distribution_DP} and noiseless privacy \cite{noiseless}, which study the relaxation of DP and take the data entropy into account. As mentioned before, the original DP definition ensures indistinguishable likelihoods of any two adjacent datasets in the worst case. Instead, \cite{distribution_DP,noiseless} considered statistical adjacent datasets where both their common part and the differing datapoint are generated from some distributions. The goal of such relaxation is to exploit the data entropy to substitute (part of) the external perturbation for the release.  However, since those works are still developed within the DP framework, \cite{distribution_DP,noiseless} encounter similar restrictions as those in generic DP analysis. First, analytical distributional DP \cite{distribution_DP} and noiseless privacy \cite{noiseless} bounds are only known for a limited number of applications with very specific assumptions on the underlying data distributions. Second, the privacy concern is still regarding the participation of an individual. In contrast, PAC Privacy can automatically handle any black-box data generation and processing mechanism for arbitrary adversarial inference. 

\noindent \textbf{Empirical Membership Inference/Data Reconstruction}: As a consequence of the loose/conservative worst-case privacy analysis such as DP \cite{dwork2006calibrating} and maximal leakage \cite{max_leakage}, a large number of works are devoted to empirically measuring the actual information leakage \cite{shokri2017membership,yeom2018privacy}  or the actual privacy guarantees produced by existing privatization methods \cite{nasr2021adversary,stock2022defending} especially in machine learning. For example, \cite{shokri2017membership} measured the influence of overfitting on the adversary’s advantage of membership or attribute inference. The results in \cite{shokri2017membership} suggest that without proper privacy preservation, many popular machine learning algorithms could have severe privacy risks. On the other hand, \cite{nasr2021adversary} empirically studied the actual privacy guarantee provided by DP-SGD and showed that there is a substantial gap between the best theoretical bound we can claim so far and the practical distinguishing advantage of the adversary. \cite{stock2022defending} studied DP-SGD from a different angle. \cite{stock2022defending} showed that compared to distinguishing individual participation, DP-SGD may provide a much stronger privacy guarantee against extraction of rare features in training data. In other words, certain data reconstruction tasks could be much harder than simply identifying the enrollment of an individual sample. This matches our earlier discussion on data reconstruction and distinguishing attack in the introduction section. Empirical works on overfitting and generalization control are complementary, and can instruct the design of more stable processing mechanisms with better privacy-utility tradeoff under PAC Privacy. 

\noindent \textbf{Fisher Privacy}: A recently proposed privacy metric, which also measures the data reconstruction error, is rooted in the Fisher information and the Cramér–Rao bound, and is termed Fisher information leakage \cite{fisher2017,fisher2022Guo,fisher2021}. Given an upper bound of Fisher information, one can lower bound the expected $l_2$ norm of reconstruction error for an adversary with {\em given bias}. However, Fisher information leakage is still input-independent and based on likelihood difference. There are several essential differences compared to priors-based PAC Privacy. First, to apply the Cramér–Rao bound in data reconstruction, the bias of the adversary's reconstruction on each selection of $X$ must be specified beforehand, while the optimal selection of such point-wise bias is, in general, unknown. In contrast, PAC Privacy does not put any restrictions or assumptions on the adversary's strategy. Second, Fisher information only measures the mean squared error whereas PAC Privacy handles impossibility of arbitrary inference tasks and criteria. 

\noindent \textbf{Generalization Error}: In learning theory, generalization error is defined as the gap between the loss (prediction accuracy) by applying the trained model on local data (seen) and on the data population (unseen) \cite{xu2017information}. To be formal, in our context, let $\mathcal{M}(X)$ represent the model learned via the (possibly randomized) algorithm $\mathcal{M}$ from a set of samples $X$, generated from some distribution $\mathsf{D}$. Given some loss $\mathcal{L}$, the generalization error is defined as 
$$ \text{Gen}(\mathsf{D}, \mathcal{M}) = \mathbb{E}_{\bar{X} \perp X \sim \mathsf{D}, \mathcal{M}}\big(\mathcal{L}(\bar{X},\mathcal{M}(X))\big) - \mathbb{E}_{X \sim \mathsf{D}, \mathcal{M}} \big( \mathcal{L}(X, \mathcal{M}(X))\big).$$
$\bar{X} \perp X \sim \mathsf{D}$ represents that $\bar{X}$ and $X$ are i.i.d. in $\mathsf{D}$. $\mathbb{E}_{\bar{X} \perp X \sim \mathsf{D}, \mathcal{M}}\big(\mathcal{L}(\bar{X},\mathcal{M}(X))\big)$ captures the expected accuracy from the model $\mathcal{M}(X)$ on population data, while $\mathbb{E}_{X \sim \mathsf{D}, \mathcal{M}} \big( \mathcal{L}(X, \mathcal{M}(X))\big)$ captures the expected training accuracy. It is well known that DP can control the generalization error \cite{bassily2016algorithmic}. Indeed, the relationship between DP, PAC Privacy and generalization error can be described as follows. As mentioned earlier, DP guarantees of $\mathcal{M}(X)$ can upper bound the mutual information $\mathsf{MI}(X; \mathcal{M}(X))$ \cite{bun2016concentrated}, and from \cite{xu2017information}, $\mathsf{MI}(X; \mathcal{M}(X))$ can then upper bound the generalization error. In the following, we show that PAC Privacy can also control the generalization error. For simplicity, we assume that the distribution $\mathsf{D}$ of $X$ is uniform over a finite set $\mathcal{X}^*=\{X_1, X_2, \ldots, X_N \}$ of $N$ elements. 

\begin{Thm}
For any given $\mathcal{M}: \mathcal{X}^* \to \mathcal{Y}$, and loss function $\mathcal{L}: (\mathcal{X}^*, \mathcal{Y}) \to (0,1)$, suppose $\mathsf{D}$ is a uniform distribution over $\mathcal{X}^*$.  If $\mathcal{M}$ satisfies $\Delta_{KL}\delta \leq v$, for the criterion $\rho(\tilde{X}, X) =1$ if $\tilde{X}= X$, i.e., the identification problem, then we have $\text{Gen}(\mathsf{D}, \mathcal{M}) \leq \sqrt{v/2}$. 
\label{thm: generalzation_error}
\end{Thm}

\begin{proof}
See Appendix \ref{app: thm: generalzation_error}. 
\end{proof}

However, we need to stress that the reverse direction does not hold in general. Recall the example of mutual information in mean estimation in Section \ref{sec:deter_alg}: we can construct some discrete distribution where given sufficiently many samples one can learn the true mean arbitrarily accurately but the mean estimator $\mathcal{M}$ could be a bijective function, where $\mathcal{M}(X)$ essentially leaks everything about $X$. However, small generalization error does help PAC Privacy if the learning mechanism ensures the trained out model is close to the population optimum. Thus, even in a case that $\mathcal{M}(\cdot)$ is bijective, we only need small noise to satisfy a mutual information bound. 

\section{Conclusions and Prospects}

\label{sec:conclusion}
In this paper, we propose and study a new instance-based information-theoretic privacy notion, termed PAC Privacy. PAC Privacy enjoys intuitive and meaningful interpretations: it enables concrete measurement on the adversary's success rate or the posterior advantage for arbitrary data inference/reconstruction task with the observation of disclosures. A simple quantification via generic $f$-divergence is presented. More importantly, based on data priors, we show an automatic privacy analysis and proof generation framework, where theoretically no algorithmic worst-case proof is needed. Data leakage control of black-box oracles becomes possible via PAC Privacy. Though the instance-based and adversarial setups are, in general, incomparable, PAC Privacy does have strong connections to cryptography and DP. Indeed, many problems or concerns in the research or applications of DP and cryptography have a corresponding version in PAC Privacy. The underlying technical challenges are not necessarily the same but motivations are very similar. We list some interesting problems here.         

\noindent \textbf{Gap between the Local and Central}: Collaboration of semi-honest users to amplify local DP via shuffling models has been recently proposed and studied in \cite{cheu2019distributed}, \cite{erlingsson2019amplification}. The key idea is to securely aggregate the local noisy response from each user, where the underlying locally added noises also aggregate to produce a more powerful perturbation with amplified privacy. As analyzed in Appendix \ref{sec:local}, via MPC, PAC Privacy analysis can be implemented in a decentralized setup and each user's local PAC Privacy would then also benefit from other people's data entropy, similar to a centralized case. However, local PAC Privacy puts more requirements on MPC, and cheap MPC implementation is a key challenge in making it practical in large-scale systems. 

\noindent \textbf{Optimal Perturbation}: Minimal perturbation is another fundamental problem in both PAC Privacy and DP.  In DP, the theoretical study of (asymptotically) optimal utility loss in different data processing remains very active. Many tight results are known for mean estimation \cite{optimal_mean}, convex ERM optimization \cite{bassily2014private} with specific solution restrictions like on $l_1$ norm \cite{l_1}, stochastic optimization \cite{feldman2020private}, and principal component analysis (PCA) \cite{PCA}. All these problems can also be systemically studied under PAC Privacy. While in PAC Privacy, privacy at a certain level could come for free (see Example \ref{ex: gap}), our current automatic framework needs noise to produce generic high-confidence security parameters. Therefore, besides a theoretically tight perturbation bound, PAC Privacy also raises its own special research problem, i.e., how to efficiently implement privacy analysis protocols and produce perturbation schemes matching the optimal bound? For example, with assistance of public data, how can we properly truncate the output domain and reduce sampling/simulation complexity? 

\noindent
\textbf{Acknowledgements:} We gratefully acknowledge the support of DSTA Singapore, Cisco Systems, Capital One, and a MathWorks fellowship. We also thank the anonymous reviewers for their helpful comments. 

\newpage

\bibliographystyle{splncs04}
\bibliography{ref.bib}

\newpage
\appendix

\begin{center}
{\Large \bf Appendix}
\end{center}

\section{Proof of Theorem \ref{thm: PAC_joint} and Corollary \ref{cor:KL-tv}}
\label{app:thm:Pac_joint}
To start, we need the following lemmas. 
\begin{Lem}
Given any $f$-divergence $\mathcal{D}_{f}(\cdot \| \cdot)$, and three Bernoulli distributions $\bm{1}_a$, $\bm{1}_b$ and $\bm{1}_c$ of parameters $a$, $b$ and $c$,  respectively, where $0\leq a \leq b \leq c \leq 1$. Then, $\mathcal{D}_{f}(\bm{1}_a \| \bm{1}_b) \leq \mathcal{D}_{f}(\bm{1}_a \| \bm{1}_c).$ 
\label{lem: f-bernoulli}
\end{Lem}
\begin{proof}
By the definition, $g(x) = \mathcal{D}_{f}(\bm{1}_a \| \bm{1}_x) = x f(\frac{a}{x}) + (1-x)f(\frac{1-a}{1-x})$ and we want to show $g(x)$ is non-decreasing for $x \geq a$. With some calculation, $g'(x) = \big(f(\frac{a}{x}) - \frac{a}{x} f'(\frac{a}{x}) \big) - \big(f(\frac{1-a}{1-x}) - \frac{1-a}{1-x} f'(\frac{1-a}{1-x}) \big).$ It is noted that $\frac{a}{x} \leq \frac{1-a}{1-x}$ for $x \geq a$. Thus, to show $g'(x) \geq 0$ for $x \geq a$, it suffices to show $t(y) = f(y) - y f'(y)$ is non-increasing with respect to $y \in [0,1]$. On the other hand, $t'(y) = f'(y)-f'(y) - yf''(y) \leq 0$ due to the convex assumption of $f$. Therefore, the claim holds.  \qedsymbol
\end{proof}

\begin{Lem}[Data Processing Inequality \cite{sason2016f}]
Consider a channel that produces $Z$ given $Y$ based on the law described as a conditional distribution $\mathsf{P}_{Z|Y}$. If $\mathsf{P}_{Z}$ is the distribution of $Z$ when $Y$ is generated by $\mathsf{P}_{Y}$, and $\mathsf{Q}_Z$ is the distribution of $Z$ when $Y$
is generated by $\mathsf{Q}_{Y}$, then for any f-divergence $\mathcal{D}_f$,
$$ \mathcal{D}_f(\mathsf{P}_{Z}\| \mathsf{Q}_{Z}) \leq \mathcal{D}_f(\mathsf{P}_{Y}\| \mathsf{Q}_{Y}). $$
\label{lem: postprocessing}
\end{Lem}

Now, we return to prove Theorem \ref{thm: PAC_joint}. First, we have the observation that for a random variable $X' \in \mathcal{X}^*$ in an arbitrary distribution but independent of $X$, $\delta^{\rho}_o \leq \Pr_{X'\perp X}\big( \rho(X', X) = 1 \big),$ since $\delta^{\rho}_o$ is the minimum failure probability achieved by optimal {\em a priori} estimation. Here, $a \perp b$ represents that $a$ is independent of $b$.  Let the indicator be a function that for two random variables $a$ and $b$, $\bm{1}(a,b)=1$ if $\rho(a,b) = 1$, otherwise $0$. Apply Lemma \ref{lem: f-bernoulli} and Lemma \ref{lem: postprocessing}, where we view $\bm{1}(\cdot, \cdot)$ as a post-processing on $(X, \tilde{X})$ and $(X, {X}')$, respectively, we have that 
$$ \mathcal{D}_{f} \big( \bm{1}_{\delta}\| \bm{1}_{\delta^{\rho}_o} \big) \leq \mathcal{D}_{f} \big( \bm{1}(X, \tilde{X}) \| \bm{1}(X, X') \big) \leq \mathcal{D}_{f} \big( \mathsf{P}_{X,\tilde{X}} \| \mathsf{P}_{X, X'} \big) = \mathcal{D}_{f} \big( \mathsf{P}_{X,\tilde{X}} \| \mathsf{P}_{X} \otimes \mathsf{P}_{X'} \big). $$
On the other hand, we know $X \to \mathcal{M}(X) \to \tilde{X}$ forms a Markov chain, where the adversary's estimation $\tilde{X}$ is dependent on observation $\mathcal{M}(X)$. Let the adversary's strategy be some operator $g_{adv}$ where $\tilde{X}= g_{adv}(\mathcal{M}(X))$. Therefore, we can apply the data processing inequality again, where 
$$ \mathcal{D}_{f} \big( \mathsf{P}_{X,\tilde{X}} \| \mathsf{P}_{X, X'} \big) \leq  \mathcal{D}_{f} \big( \mathsf{P}_{X,\mathcal{M}(X)} \| \mathsf{P}_{X, W} \big) = \mathcal{D}_{f} \big( \mathsf{P}_{X,\mathcal{M}(X)} \| \mathsf{P}_{X} \otimes \mathsf{P}_{W} \big).$$
Here, $X' = g_{adv}(W)$ and $W$ is still independent of $X$. Since the above inequalities hold for arbitrarily distributed $X'$ once it is independent of $X$, $W$ could also be an arbitrary random variable on the same support domain as $\mathcal{M}(X)$ and independent of $X$. Therefore,
$$ \mathcal{D}_{f} \big( \bm{1}_{\delta} \| \bm{1}_{\delta^{\rho}_o} \big) \leq \inf_{\mathsf{P}_W} \mathcal{D}_{f}\big( \mathsf{P}_{X, \mathcal{M}(X)} \| \mathsf{P}_{X} \otimes \mathsf{P}_W \big) = \inf_{\mathsf{P}_W} \mathcal{D}_{f}\big( \mathsf{P}_{\mathcal{M}(X)|X} \| \mathsf{P}_{W} |\mathsf{P}_X).$$
Here, we use $|\mathsf{P}_X$ to denote that it is conditional on $X$ in a distribution $\mathsf{P}_{X}$.  
In particular, if we select $\mathsf{P}_W$ to be the distribution of $\mathcal{M}(X)$, and take $\mathsf{D}_f$ to be KL-divergence, we have  $\mathcal{D}_{KL} \big( \bm{1}_{\delta} \| \bm{1}_{\delta^{\rho}_o} \big) \leq \mathsf{MI}(X; \mathcal{M}(X)).$

\section{Proof of Theorem \ref{thm:i.i.d.pac}} 
\label{app:thm:i.i.d.pac}
The first part of Theorem \ref{thm:i.i.d.pac} is a simple corollary of Theorem \ref{thm: PAC_joint}. When we separately consider the adversary's inference on each $x_i$, take $x_i$ as the $X$ in the proof of Theorem \ref{thm: PAC_joint} and a similar reasoning will produce the upper bound on each posterior success rate $(1-\delta_i)$, and the claim follows. 

Consider the second part of Theorem \ref{thm:i.i.d.pac}, where we assume the $x_i$ is i.i.d. generated and $\mathcal{M}$ is some symmetric mechanism. We apply Theorem \ref{thm: PAC_joint} to upper bound the success rate that the adversary can recover at least $j$ many $x_i$s out of the total $X$, as some generic inference on the whole data $X$. Due to the i.i.d. assumption, we know the optimal {\em a priori} success rate to recover a single $x_i$ equals $1-\bar{\delta}^{\rho}_o = \sup_{\tilde{x}}\Pr(\rho( \tilde{x}, x) =1).$ Consequently, the optimal {\em a priori} success rate to recover at least $j$ out of $n$ samples is $\tilde{\delta}^{j,\rho}_{o}= \sum_{l=j}^n {n \choose l} (\bar{\delta}^{\rho}_{o})^{n-l} \cdot (1-\bar{\delta}^{\rho}_{o})^{l}.$  Thus, applying Theorem \ref{thm: PAC_joint}, we have that the optimal posterior success rate on recovering at least $j$ samples $(1-\tilde{\delta}^j)$ satisfies
$$ \mathcal{D}_{KL}(\bm{1}_{\tilde{\delta}^j}\|\bm{1}_{\tilde{\delta}^{j, \rho}_{o}}) \leq \mathsf{MI}\big(X;\mathcal{M}(X)\big).$$ 
On the other hand, due to the i.i.d. and the symmetric assumption of $\mathcal{M}$, we know $\mathsf{MI}(x_i; \mathcal{M}(X))$ is identical, for any $i \in [1:n]$. The optimal posterior success rate to recover every single $x_i$ is the same as $(1-\delta)$. Therefore, let the indicator $\bm{1}_{x_i}$ represent whether $x_i$ is successfully recovered. The expected number of samples that an adversary can recover satisfies 
$$ n(1-\delta) = \mathbb{E}[\sum_{i=1}^n \bm{1}_{x_i}] \leq \sum_{j=1}^n (1-\tilde{\delta}^j).$$
Here, we use Fubini's theorem on the expectation calculation, where for a random variable $Y$ defined over $\mathbb{Z}^+$, 
$$ \mathbb{E}[Y] = \sum_{j=1}^{+\infty} \Pr(Y \geq j).$$
Thus, we obtain the desired lower bound of $\delta$ expressed by $\tilde{\delta}^j$ for $j=1,2,\ldots,n$. 

\section{Proof of Theorem \ref{thm:deter_m}}
\label{app:thm:deter_m}
\begin{equation}
   \begin{aligned}
    &\mathsf{MI} (X; \mathcal{M}(X)+B)\\
    &=  \int \mathcal{D}_{KL}(\mathsf{P}_{\mathcal{M}(X_0)+B} \| \mathsf{P}_{B}) \mathsf{P}(X=X_0)~ dX_0 - \mathcal{D}_{KL}(\mathsf{P}_{\mathcal{M}(X)+B}\| \mathsf{P}_{B}) \\
    & = \int \mathcal{D}_{KL}(\mathsf{P}_{\mathcal{M}(X_0)+B} \| \mathsf{P}_{B}) \mathsf{P}(X=X_0)~ dX_0 - \big(\mathcal{D}_{KL}(\mathsf{P}_{\mathcal{M}(X)+B}\| \mathsf{P}_{Gau(\mathcal{M}(X)+B)}) + \mathcal{D}_{KL}(\mathsf{P}_{Gau(\mathcal{M}(X)+B)}\|\mathsf{P}_{B}) \big) \\
    & \leq \int \mathcal{D}_{KL}(\mathsf{P}_{\mathcal{M}(X_0)+B} \| \mathsf{P}_{B}) \mathsf{P}(X=X_0)~ dX_0 - \mathcal{D}_{KL}(\mathsf{P}_{Gau(\mathcal{M}(X)+B)}\| \mathsf{P}_{B}).
   \end{aligned}
\label{Gau_approx}
\end{equation}
Given the definition of mutual information, we first apply the results of Gaussian approximation \cite{pinsker1995sensitivity}, where $Gau(A)$ represents a (multivariate) Gaussian variable with the same mean and (co)variance as those of $A$. Then, we drop a negative term to obtain the final inequality of (\ref{Gau_approx}). Next,
focusing on the integral (first) term in (\ref{Gau_approx}), 
$$ \mathcal{D}_{KL}(\mathsf{P}_{\mathcal{M}(X_0)+B} \| \mathsf{P}_{B}) = \frac{1}{2} \cdot (\mathcal{M}(X_0))^T\Sigma^{-1}_{\bm{B}}(\mathcal{M}(X_0)), $$
and therefore 
$$ \int \mathcal{D}_{KL}(\mathsf{P}_{\mathcal{M}(X_0)+B} \| \mathsf{P}_{B}) \mathsf{P}(X=X_0)~ dX_0  =  \frac{1}{2} \cdot \mathbb{E}_{X}\big[(\mathcal{M}(X))^T\Sigma^{-1}_{\bm{B}}(\mathcal{M}(X))\big].$$
As for the second term in the last equation of (\ref{Gau_approx}), we have the covariance of $\mathcal{M}(X)$ equals $\Sigma_{\mathcal{M}(X)} = \mathbb{E}_{X}\big[(\mathcal{M}(X)-\mathbb{E}[\mathcal{M}(X)] )(\mathcal{M}(X)-\mathbb{E}[\mathcal{M}(X)])^T \big]$, while the mean $\mu_{\mathcal{M}(X)} =\mathbb{E}[\mathcal{M}(X)]$. The KL divergence between two multivariate Gaussians has a closed form, where $\mathcal{D}_{KL}(\mathsf{P}_{Gau(\mathcal{M}(X)+B)}\| \mathsf{P}_{B})$ equals 
\begin{equation}
\begin{aligned}
 \mathcal{D}_{KL}&(\mathsf{P}_{Gau(\mathcal{M}(X)+B)}\| \mathsf{P}_{B})  = \frac{1}{2} \cdot \big( \text{Trace}(\Sigma_{\mathcal{M}(X)}\cdot \Sigma^{-1}_{\bm{B}})\\
 & +\mathbb{E}_{X}[\mathcal{M}(X)]^T\Sigma^{-1}_{\bm{B}}\mathbb{E}_{X}[\mathcal{M}(X)]- \log \text{det}(\bm{I}_{d} + \Sigma_{\mathcal{M}(X)}\Sigma^{-1}_{\bm{B}}) \big).
\end{aligned}
\end{equation}
On the other hand, note that 
\begin{equation}
\begin{aligned}
  & \mathbb{E}_{X}\big[(\mathcal{M}(X))^T\Sigma^{-1}_{\bm{B}}(\mathcal{M}(X))\big] - \mathbb{E}_{X}\big[\mathcal{M}(X)\big]^T\Sigma^{-1}_{\bm{B}}\mathbb{E}_{X}\big[\mathcal{M}(X)\big] - \text{Trace}(\Sigma_{\mathcal{M}(X)}\cdot {\Sigma}^{-1}_{\bm{B}})\\
  & = \text{Trace}\big( \mathbb{E}[\mathcal{M}(X)-\mathbb{E}[\mathcal{M}(X)]] \cdot  \Sigma^{-1}_{\bm{B}} \cdot \mathbb{E}[\mathcal{M}(X)-\mathbb{E}[\mathcal{M}(X)]]^T \big)-\text{Trace}\big(\Sigma_{\mathcal{M}(X)}\Sigma^{-1}_{\bm{B}} \big) \\
  & = \text{Trace}\big( \mathbb{E}[\mathcal{M}(X)-\mathbb{E}[\mathcal{M}(X)]] \cdot \mathbb{E}[\mathcal{M}(X)-\mathbb{E}[\mathcal{M}(X)]]^T \cdot \Sigma^{-1}_{\bm{B}}- \Sigma_{\mathcal{M}(X)}\Sigma^{-1}_{\bm{B}} \big) \\
  & =\text{Trace}\big( \Sigma_{\mathcal{M}(X)}\cdot \Sigma^{-1}_{\bm{B}}- \Sigma_{\mathcal{M}(X)}\Sigma^{-1}_{\bm{B}} \big) = 0.
\end{aligned}
\label{trace bound}
\end{equation}
In (\ref{trace bound}), we use the following facts that for two arbitrary vectors $v_1, v_2 \in \mathbb{R}^d$, $(v_1)^Tv_2 = \text{Trace}(v_1\cdot (v_2)^T)$, and for two arbitrary  matrices $A_1, A_2 \in \mathbb{R}^{d \times d}$, $\text{Trace}(A_1A_2)=\text{Trace}(A_2A_1)$. Therefore, putting it all together, we have a simplified form of the right hand of (\ref{Gau_approx}), where 
\begin{equation}
    \label{MI_bound_log}
    \mathsf{MI} (X; \mathcal{M}(X)+B) \leq \frac{\log \text{det}(\bm{I}_{d}+ \Sigma_{\mathcal{M}(X)}\cdot \Sigma^{-1}_{\bm{B}})}{2}.
\end{equation}

In the following, we turn to find a construction of $\Sigma_{\bm{B}}$ to control the right side of (\ref{MI_bound_log}). Assume the SVD of $\Sigma_{\mathcal{M}(X)}=U\Lambda U^T$, where $\Lambda$ represents a diagonal matrix whose diagonal elements are the eigenvalues $\lambda_1, \ldots , \lambda_d$ of $\Sigma_{\mathcal{M}(X)}$. We take $\Sigma_{\bm{B}} = U\Lambda_{\bm{B}}U^T$, where $\Lambda_{\bm{B}}$ is a diagonal matrix of elements $\lambda_{\bm{B},1}, \ldots ,\lambda_{\bm{B},d}$. Using the fact that $\log(1+x)\leq x$, we can further upper bound the right side of (\ref{MI_bound_log}) and obtain,
$$  \mathsf{MI} (X; \mathcal{M}(X)+B) \leq \frac{1}{2} \cdot \sum_{j=1}^d \frac{\lambda_j}{\lambda_{\bm{B},j}}.$$
On the other hand, note that $\mathbb{E}[\|\bm{B}\|^2_2] = \sum_{j=1}^d \lambda_{\bm{B},j}$. Therefore, to find the minimal noise, it is equivalent to solving the following optimization problem.
\begin{equation}
\label{min_noise}
\begin{aligned}
 \min \sum_{j=1}^d \lambda_{\bm{B},j}, \quad \quad  s.t. \sum_{j=1}^d \frac{\lambda_j}{\lambda_{\bm{B},j}} \leq 2v,  \lambda_{\bm{B},j} \geq 0, j=1,2,\ldots,d.
\end{aligned}
\end{equation}
By Hölder's inequality on
$$ \sum_{j=1}^d \lambda_{\bm{B},j} \cdot \sum_{j=1}^d \frac{\lambda_j}{\lambda_{\bm{B},j}} \leq 2v \cdot \sum_{j=1}^d \lambda_{\bm{B},j},$$
the optimal solution of (\ref{min_noise}) is $\lambda_{\bm{B},j} = (2v)^{-1}\cdot \sqrt{\lambda_j}\cdot(\sum_{l=1}^d \sqrt{\lambda_l}). $ As a consequence, the expected $l_2$ norm of the noise $\mathbb{E}[\|\bm{B}\|_2]$ is upper bounded by  $\sqrt{(2v)^{-1}}\cdot (\sum_{j=1}^d \sqrt{\lambda_j}).$

\section{Proof of Theorem \ref{thm: deter_alg}}
\label{app:proof_deter_alg}
We first consider two arbitrary positive semi-definite $\Sigma$ and $\hat{\Sigma}$, where $\Delta\Sigma = \Sigma-\hat{\Sigma}$ denotes the difference. Let their SVD decompositions be $\Sigma= U\Lambda U^T= \sum_{j=1}^d \lambda_i  u_i u^T_i$ and $\hat{\Sigma}= \hat{U}\hat{\Lambda}\hat{U}^T = \hat{\lambda}_i\hat{u}_i \hat{u}^T_i$. Before the proof of Theorem \ref{thm: deter_alg}, we need the following lemmas.

\begin{Lem}[Weyl's theorem]
For two $d \times d$ positive semi-definite matrices $\Sigma$ and $\hat{\Sigma}$, let their eigenvalues be $\lambda_{1} \geq \ldots \geq \lambda_{d}$ and $\hat{\lambda}_{1} \geq \ldots \geq \hat{\lambda}_{d}$, both in a non-increasing order. Then, for any $1 \leq j \leq d$,
$$ |\lambda_{j} - \hat{\lambda}_{j}| \leq \|\Sigma-\hat{\Sigma}\|_2.$$
\label{lem: weyl}
\end{Lem}

\begin{Lem}
Assume that $Y = \{y_1, y_2, \ldots, y_m\}$ is i.i.d. selected from some distribution $\mathsf{D}_y$ whose support domain is $\mathcal{B}^d_{r}$ and the covariance matrix is $\Sigma$. Let $\hat{\Sigma}$ be the empirical covariance of $Y$, then we have that with probability at least $(1-\gamma)$,  $\|\Delta\Sigma= \hat{\Sigma} - \Sigma\|$ is upper bounded by 
$$\kappa r\big(\max\{ \sqrt{\frac{d+\log(4/\gamma)}{m}},\frac{d+\log(4/\gamma)}{m}\} + \sqrt{\frac{d\log(4/\gamma)}{m}} \big).$$
\end{Lem}

\begin{proof}
We will use some notions and the high-dimensional statistics results from \cite{vershynin2018high}. First, we have the following w.r.t. the empirical covariance. Let $\hat{y}= \frac{1}{m}\cdot \sum_{i=1}^m y_i$, and the empirical covariance is defined as $\hat{\Sigma} = \frac{1}{m} \sum_{i=1}^m (y_i-\hat{y})(y_i-\hat{y})^T.$ Then, for any $\mu \in \mathbb{R}^d$, we have 
$$ \frac{1}{m}\cdot \sum_{i=1}^m (y_i-\mu)(y_i-\mu)^T = \hat{\Sigma} + (\mu-\hat{y})(\mu-\hat{y})^T. $$
Second, since $y \sim \mathsf{D}_y$ is uniformly bounded and we know $y$ is some $r$-subGaussian random vector (see Def 3.4.1 in \cite{vershynin2018high}). Now, we select $\mu = \mathbb{E}[y]$, and for zero-mean $r$-sub-Gaussian random vectors, from Thm 4.7.1 and Ex 4.7.3 in \cite{vershynin2018high}, we have that there exists some universal constant $\kappa'$ such that with probability at least $(1-\gamma')$,  
$$ \|\Sigma -\frac{1}{m}\cdot \sum_{i=1}^m (y_i-\mu)(y_i-\mu)^T \| \leq \kappa' r \cdot \max\{ \sqrt{\frac{d\log(2/\gamma')}{m}}, \frac{d\log(2/\gamma')}{m}\}.$$
On the other hand, we provide a high-probability concentration on the mean estimation $\|\mu - \hat{y}\|.$ We use the results on {\em norm-subGaussian} \cite{jin2019short}, where 
$$ \Pr(\|\mu - \hat{y}\| \geq t) \leq 2e^{-\frac{mt^2}{2dr^2}}.$$
Therefore, with probability at least $(1-\gamma')$, $\|(\mathbb{E}[y]-\hat{y})(\mathbb{E}[y]-\hat{y})^T\| \leq \sqrt{\frac{-2dr^2\log(\gamma'/2)}{m}}.$ Now, let $\gamma' = \frac{\gamma}{2}$, and put things together, we have that 
there exists some universal constant $\kappa$,
$$ \|\hat{\Sigma} - \Sigma\| \leq \kappa r \big(\max\{ \sqrt{\frac{d+\log(4/\gamma)}{m}},\frac{d+\log(4/\gamma)}{m}\} + \sqrt{\frac{d\log(4/\gamma)}{m}} \big). $$ \qedsymbol
\end{proof}

\begin{Lem}
Let $\hat{\Sigma}= \Sigma+ \Delta\Sigma = \hat{U}\hat{\Lambda}\hat{U}^T$ for $\hat{\Lambda}=\{ \hat{\lambda}_1, \ldots , \hat{\lambda}_d\}$. If $\|\Delta\Sigma\| \leq c$, then for some particular $j \in [1:d]$, if the gap $\min_{l \not= j} |\lambda_j-\lambda_l| \geq \tau$, then  
$$ \sum_{l \not= j} \lambda_l (u^T_l \cdot \hat{u}_j)^2 \leq  4c + 4\lambda_j\min\big\{1, (d-1)c^2/\tau^2 \big\}.$$
\label{lem: subspace_approx}
\end{Lem}
\begin{proof}
For $\hat{\Sigma}$, we can rewrite its $j$-th eigenvector as the following
\begin{equation}
    \begin{aligned}
 \hat{\Sigma}\hat{u}_j & = (\Sigma+\Delta\Sigma)(u_j + \Delta u_j) & = \hat{\lambda}_j \hat{u}_j 
 = (\lambda_j+\Delta\lambda_j)(u_j + \Delta u_j). 
    \end{aligned}
\label{basic_eigen_difference}
\end{equation}
From (\ref{basic_eigen_difference}), since $\Sigma u_j = \lambda_j u_j$, we have that 
\begin{equation}
    \Sigma\Delta u_j = (\lambda_j+\Delta\lambda_j)(u_j + \Delta u_j) - \lambda_ju_j - \Delta\Sigma(u_j + \Delta u_j).
\label{basic_eigen_difference 1}
\end{equation}
We multiply $(\Delta u_j)^T$ on both sides of (\ref{basic_eigen_difference 1}) and have that 
\begin{equation}
\begin{aligned}
(\Delta u_j)^T\Sigma\Delta u_j = \lambda_j \|\Delta u_j \|^2 - (\Delta u_j)^T(\Delta\Sigma -\Delta \lambda_i \cdot \bm{I}_d) \cdot (\hat{u}_j).  
\end{aligned}
\label{basic_eigen_difference_2}
\end{equation}
Now, we apply the equivalent expression $\Sigma = \sum_{l=1}^d \lambda_l u^T_lu_l$ in (\ref{basic_eigen_difference_2}) and have
\begin{equation}
    \begin{aligned}
      \lambda_j(u^T_j \cdot \Delta u_j)^2  +  \sum_{l\not=j} \lambda_l (u^T_l \cdot \Delta {u}_j)^2 & \leq \lambda_j \|\Delta u_j \|^2 + 2\| \Delta u_j\| \|\Delta\Sigma \| \| \hat{u}_j\|\\
      & \leq \lambda_j \|\Delta u_j \|^2 + 4\|\Delta\Sigma \|.
    \end{aligned}
\label{basic_eigen_difference_3} 
\end{equation}
Here, we use the following results. First, from Weyl's theorem on eigenvalues, $|\Delta\lambda_j| = |\lambda_j-\hat{\lambda}_j| \leq \|\Delta\Sigma\|$ for any $j$. Second, $\|\Delta u_j\| = \|u_j-\hat{u}_j\| \leq 2$.  

Now, we consider multiplying $u^T_l$ on both sides of \ref{basic_eigen_difference 1} and we have the following,
$$\lambda_l (u^T_l\cdot \Delta u_j) = \lambda_j (u^T_l\cdot \Delta u_j) - u^T_l(\Delta\Sigma -\Delta \lambda_j \cdot \bm{I}_d) \cdot (\hat{u}_j).$$
Rearrange the above and take absolute value of both sides, then we have for any $l \not=j$, 
$$ |u^T_l\cdot \Delta u_j| \leq 2\|\Delta\Sigma\|/|\lambda_l-\lambda_j| \leq 2\|\Delta\Sigma\|/\tau.$$
With the above preparation, we can proceed to upper bound $\|\Delta u_j\|^2-(u^T_j\cdot \Delta u_j)^2$ as follows. Since $u_{[1:d]}$ are a set of orthogonal basis of $\mathbb{R}^d$, for a given $j$, 
\begin{equation}
    \begin{aligned}
      \|\Delta u_j\|^2-(u^T_j\cdot \Delta u_j)^2 &= \sum_{l \not= j} (u^T_l \cdot \Delta u_j)^2 &\leq 4(d-1) \|\Delta \Sigma\|^2/\tau^2. 
    \end{aligned}
\label{delta_u}
\end{equation}
Now we put (\ref{delta_u}) back into (\ref{basic_eigen_difference_3}) and obtain that
$$ \sum_{l\not=j} \lambda_l (u^T_l \cdot \Delta {u}_j)^2 \leq 4\big(\lambda_j(d-1)\|\Delta\Sigma\|^2/\tau^2+ \|\Delta\Sigma\| \big).$$
It is noted that 
$$ \sum_{l \not= j} \lambda_l (u^T_l \cdot \hat{u}_j)^2 = \sum_{l \not=j} \lambda_l (u^T_l \cdot \Delta u_j)^2,$$
since $\langle u_l, u_j\rangle=0$ for $l \not =j$, and on the other hand we know $\lambda_j\|\Delta u_j\|^2 \leq 4\lambda_j$, and therefore we obtain the upper bound claimed straightforwardly. 
\end{proof}

Back to the proof of Theorem \ref{thm: deter_alg}. First, from Theorem \ref{thm:deter_m}, we have  
 \begin{align}
 & \mathsf{MI} (X; \mathcal{M}(X)+\bm{B}) \leq \frac{1}{2} \cdot \log \text{det}\big(\bm{I}_{d'}+ \Sigma_{\mathcal{M}(X)}\cdot \Sigma^{-1}_{\bm{B}}\big) \leq \frac{1}{2} \cdot \Tr(\Sigma_{\mathcal{M}(X)}\cdot \Sigma^{-1}_{\bm{B}}). 
 \label{1/2 trace}
 \end{align}
Now, let $\Sigma_{\mathcal{M}(X)}= U^T\Lambda U$ and $\Sigma^{-1}_{\bm{B}} = \hat{U}^T \hat{\Lambda}\hat{U}$. Similarly, we follow the notations used in Lemma \ref{lem: subspace_approx} and let $\hat{u}_j= u_j + \Delta u_j. $ Then, for the first scenario defined in Algorithm \ref{alg: deter_MI}, let $j_0= \arg \max_{j} \lambda_j \geq c_0$, $\min_{ 1\leq j \leq j_0, 1 \leq l \leq d}  |\hat{\lambda}_j-\hat{\lambda}_l| > c_1$ and $\lambda_{\bm{B},j} = \frac{c_2}{\sqrt{\hat{\lambda}_j+c_3} \cdot(\sum_{j=1}^d \sqrt{\hat{\lambda}_j+c_3})},$ for parameters $c_0$, $c_1$, $c_{2}$ and $c_3$ to be determined, and we assume that $\|\Sigma - \hat{\Sigma}\| \leq c$. Then, we expand the term  $\Tr(\Sigma_{\mathcal{M}(X)}\cdot \Sigma^{-1}_{\bm{B}})$ as follows,
\begin{equation}
\begin{aligned}
  & \Tr(\Sigma_{\mathcal{M}(X)}\cdot \Sigma^{-1}_{\bm{B}}) = \Tr\big( (\sum_{j=1}^d \lambda_j u_ju^T_j) \cdot (\sum_{l=1}^d {\lambda}_{\bm{B},l} \hat{u}_lu^T_l) \big) \\
  & = \Tr\big( \sum_{j=1}^d\sum_{l=1}^d \lambda_j{\lambda}_{\bm{B},l} \cdot  (u^T_j\hat{u}_l )^2 \big) = \Tr\big( \sum_{j=1}^d\sum_{l=1}^d \lambda_j{\lambda}_{\bm{B},l} \cdot  ( u^T_j({u}_l + \Delta u_l))^2 \big) \\
  & \leq \sum_{j=1}^d \lambda_j{\lambda}_{\bm{B},j} + \sum_{j=1}^d\sum_{l \not=j} \lambda_j{\lambda}_{\bm{B},l} (  u^T_j \Delta u_l)^2  \\
  & \leq  \sum_{j=1}^d \frac{c_2\lambda_j}{\sqrt{\hat{\lambda}_j+c_3} \cdot (\sum_{j=1}^d \sqrt{\hat{\lambda}_j + c_3})} 
    + \sum_{j=1}^{d} \frac{c_2}{d c_3} (4c + 4\lambda_j\min\big\{1, (d-1)c^2/\tau^2_j \big\})\\
  & \leq \sum_{j=1}^d \frac{c_2\lambda_j}{\sqrt{\hat{\lambda}_j+c_3} \cdot (\sum_{j=1}^d \sqrt{\hat{\lambda}_j + c_3})} + \frac{j_0(4c+16dc^2r^2/(c_1-2c)^2)c_2+(d-j_0)(4c+4c_0)c_2}{d{c_3}}.
\end{aligned}
\label{case_1_MI} 
\end{equation}
Here, in the second inequality, we use the fact that $\lambda_{\bm{B},j} \leq \frac{c_2}{dc_3}$. In the third inequality, we use the fact that $\lambda_j \leq 4r^2$. For $\tau_j = \min_{ 1\leq l \leq d} |\lambda_{j}-\lambda_{l}|$, via Weyl's Lemma (Lemma \ref{lem: weyl}), $|\lambda_j - \hat{\lambda}_j| \leq c$ and thus for $j \in [1:j_0]$, $\tau_j \geq c_1-2c$, and for any $c_3 \geq c$ we have $$\sum_{j=1}^d \frac{c_2\lambda_j}{\sqrt{\hat{\lambda}_j+c_3} \cdot (\sum_{j=1}^d \sqrt{\hat{\lambda}_j + c_3})} \leq c_2. 
$$
Now, we select $c_0=c$, $c_1=r\sqrt{dc}+2c$, $c_2=2v$, $c_3=10 cv/\beta$, and we can upper bound the right hand of (\ref{case_1_MI}) as $2(v+\beta)$. Thus, by  
(\ref{1/2 trace}), the objective $\mathsf{MI} (X; \mathcal{M}(X)+\bm{B}) \leq v+\beta$. 

The second case in Algorithm \ref{alg: deter_MI} is straightforward to analyze, where 
\begin{align*}
  \Tr(\Sigma\cdot \Sigma^{-1}_{B}) = \frac{c_2 \sum_{j=1}^d \lambda_j}{\sum_{j=1}^d \hat{\lambda}_j+dc} \leq c_2.    
\end{align*}
Here, we still apply Weyl's lemma where when $\|\hat{\Sigma}-\Sigma\|_2 \leq c$, then $|\Tr(\hat{\Sigma}) - \Tr(\Sigma)| \leq dc$. 

\section{Proof of Theorem \ref{thm: random_alg}}
\label{app:thm:random_alg}
First, regarding $\mathsf{MI}(X; \mathcal{M}(X, \theta)+B)$ we have the following observation,
\begin{equation}
    \begin{aligned}
      &\mathsf{MI}(X; \mathcal{M}(X, \theta)+B)  \\ 
      & = \sum_{X_0} \Pr(X=X_0) \cdot \mathcal{D}_{KL}(\mathsf{P}_{\mathcal{M}(X_0,\theta)+B} \| \mathsf{P}_{\mathcal{M}(X,\theta)+B} ). 
    \end{aligned}
\end{equation}
Based on our assumptions, $\mathcal{M}(X_0,\theta)+B$ is distributed as a Gaussian mixture, where for some Gaussian $B \sim \mathcal{N}(\bm{0}, \sigma^2 \cdot \bm{I}_d)$, $\mathsf{P}_{\mathcal{M}(X_0,\theta)+B}$ can be equivalently written as 
$$ \mathsf{P}_{\mathcal{M}(X_0,\theta)+B} = \sum_{j=1}^{|\bm{\Theta}|} \frac{1}{|\bm{\Theta}|} \cdot \mathcal{N}(\mathcal{M}(X_0,\theta_j), \sigma^2 \cdot \bm{I}_d). $$
Similarly, we can write $\mathsf{P}_{\mathcal{M}(X,\theta)+B}$ in the following equivalent form,
$$ \mathsf{P}_{\mathcal{M}(X,\theta)+B} = \sum_{X'_0}\sum_{j=1}^{|\bm{\Theta}|} \frac{1}{|\bm{\Theta}|} \cdot  \mathbb{P}_{X'_0}  \cdot \mathcal{N}(\mathcal{M}(X'_0,\theta_j), \sigma^2 \cdot \bm{I}_d).$$
Applying the convexity of KL-divergence, we have that 
\begin{equation}
    \begin{aligned}
      & \mathcal{D}_{KL}(\mathsf{P}_{\mathcal{M}(X_0,\theta)+B} \| \mathsf{P}_{\mathcal{M}(X,\theta)+B} ) \leq  \sum_{X'_0} \mathbb{P}_{X'_0} \cdot \mathcal{D}_{KL}(\mathsf{P}_{\mathcal{M}(X_0,\theta)+B} \| \mathsf{P}_{\mathcal{M}(X'_0,\theta)+B} )                                     \\
      & \leq \sum_{X'_0}\sum_{j=1}^{|\bm{\Theta}|}\frac{1}{|\bm{\Theta}|} \cdot \mathbb{P}_{X'_0}  \cdot \mathcal{D}_{KL}(\mathsf{P}_{\mathcal{M}(X_0,\theta_j)+B} \| \mathsf{P}_{\mathcal{M}(X'_0, \theta_{\pi(X'_0,j)})+B}) \\
      & = \sum_{X'_0}\sum_{j=1}^{|\bm{\Theta}|}\frac{1}{|\bm{\Theta}|} \cdot \mathbb{P}_{X'_0}  \cdot \frac{\|\mathcal{M}(X_0,\theta_j) - \mathcal{M}(X'_0,\theta_{\pi(X'_0, j)})  \|^2}{2\sigma^2},
    \end{aligned}
\label{on-average-KL}
\end{equation}
for any permutation $\pi_{X'_0}$ on $[1:|\bm{\Theta}|]$ for each $X'_0$. Here, we use $\mathbb{P}_{X'_0} = \Pr(X = X'_0)$. Therefore, if we let $\psi(X_0, X'_0) = \mathsf{d}_{\pi}(\bm{y}(X_0), \bm{y}(X'_0))$, where $$\bm{y}(X_0) = (\mathcal{M}(X_0,\theta_1), \ldots , \mathcal{M}(X_0,\theta_{|\bm{\Theta}|})), \bm{y}(X'_0) = (\mathcal{M}(X'_0,\theta_1), \ldots , \mathcal{M}(X'_0,\theta_{|\bm{\Theta}|})),$$ then we have a simplified upper bound on $\mathsf{MI}(X; \mathcal{M}(X, \theta)+B)$ as follows,
\begin{equation}
\label{MI_variance_control}
   \mathsf{MI}(X; \mathcal{M}(X, \theta)+B) = \mathbb{E}_{X_0}\mathcal{D}_{KL}(\mathsf{P}_{\mathcal{M}(X_0,\theta)+B} \| \mathsf{P}_{\mathcal{M}(X,\theta)+B} ) \leq \mathbb{E}_{X_0, X'_0}  \frac{\psi(X_0,X'_0)}{2\sigma^2}. 
\end{equation}
Therefore, to derive the upper bound on the objective mutual information, it suffices to estimate the expectation $\mathbb{E}_{X_0, X'_0}  \frac{\psi(X_0,X'_0)}{2\sigma^2}$, which is the mean of the permutation distance between the evaluations of $\mathcal{M}(\cdot, \theta)$ on two randomly selected input datasets $X_0$ and $X'_0$ from distribution $\mathsf{D}$. On one hand, we know $\psi(X_0,X'_0) \leq 4r^2$ due to the bounded output assumption for any $X_0$ and $X'_0$. In the following, we prove the following lemma,
\begin{Lem}
Let $\bm{\Theta}_{\tau} = \{\theta_{(1)}, \theta_{(2)}, \ldots ,\theta_{(\tau)}\}$ be a random $\tau$-subset from $\bm{\Theta}$ for some $\tau$ that can divide $|\bm{\Theta}|$. We define $\psi(X_0, X'_0, \bm{\Theta}_{\tau})$ to be permutation distance between $\big( \mathcal{M}(X_0,\theta_{(1)}), \ldots ,\mathcal{M}(X_0,\theta_{(\tau)})\big)$ and $\big( \mathcal{M}(X'_0,\theta_{(1)}), \ldots ,$ $\mathcal{M}(X'_0,\theta_{(\tau)})\big)$ restricted to the evaluations on $\bm{\Theta}_{\tau}$. Then,
$$ \mathbb{E}_{\bm{\Theta}_{\tau}} \psi(X_0, X'_0, \bm{\Theta}_{\tau}) \geq \psi(X_0, X'_0).  $$
\label{lem:matrix_subsampling}
\end{Lem}
\begin{proof}
We consider a random permutation $\pi$ on $[1:|\bm{\Theta}|]$ and we split $$\big(\mathcal{M}(X_0, \theta_{\pi(1)}), \ldots , \mathcal{M}(X_0, \theta_{\pi(|\bm{\Theta}|)})\big)$$ into $|\bm{\Theta}|/\tau$ segments in order. For the $j$-th segment, we consider the minimal permutation distance between the two sets $\big(\mathcal{M}(X_0, \theta_{\pi((j-1)\tau+1)}), \ldots , \mathcal{M}(X_0, \theta_{\pi(j\tau)})\big),$ and $\big(\mathcal{M}(X'_0, \theta_{\pi((j-1)\tau+1)}), \ldots , \mathcal{M}(X'_0, \theta_{\pi(j\tau)})\big)$, denoted by $\psi_j(X_0, X'_0, \pi)$. Since $\psi(X_0, X'_0)$ is the global minimum, it is clear that 
$$  \psi(X_0, X'_0) \leq \sum_{j=1}^{|\bm{\Theta}|/\tau}  \frac{\psi_j(X_0, X'_0, \pi)}{|\bm{\Theta}|/\tau}.$$ 
When we take the expectation on both sides, we have that
$$ \psi(X_0, X'_0) \leq  \mathbb{E}_{\bm{\Theta}_{\tau}} \psi(X_0, X'_0, \bm{\Theta}_\tau). $$ \qedsymbol
\end{proof}

With Lemma \ref{lem:matrix_subsampling}, we can proceed to show a high-probability mean estimation of  $\mathbb{E}_{X_0, X'_0}\mathbb{E}_{\bm{\Theta}_q} [{\psi(X_0,X'_0, \bm{\Theta}_q)}]$. By Hoeffding's inequality, since $\psi(X_0,X'_0, \bm{\Theta}_q) \leq 4r^2$, and let $\bar{\psi}$ be the empirical mean of $m$ i.i.d. sampling on $\psi(X_0,X'_0, \bm{\Theta}_q)$, then for any $t>0$,  
$$ \Pr(\mathbb{E}_{X_0, X'_0}\mathbb{E}_{\bm{\Theta}_q} [{\psi(X_0,X'_0, \bm{\Theta}_q)}] > \bar{\psi}+t) \leq e^{- \frac{2mt^2}{16r^4} }.$$
Thus, when $ m \geq \frac{-8r^4\log\gamma}{c^2},$ with confidence $(1-\gamma)$, one can ensure that $$ \mathbb{E}_{X_0, X'_0} {\psi(X_0, X'_0)} \leq \bar{\psi}+c. $$
As a consequence, by selecting $B \sim \mathcal{N}(0, \sigma^2\cdot \bm{I}_d)$, for $\sigma^2 = (\bar{\psi}+c)/(2v)$, we have that $\mathsf{MI}(X;\mathcal{M}(X)+\bm{B}) \leq v $ with confidence $(1-\gamma)$. 

\section{Proof of Theorem \ref{thm: verification}}
\label{app:thm: verification}
The proof of Theorem  \ref{thm: verification} is essentially a generalization of Lemma  \ref{lem:matrix_subsampling}. We rewrite the expression of $\mathsf{MI}(X; \mathcal{M}(X))$ as follows,
\begin{equation}
    \begin{aligned}
 \frac{1}{N}\sum_{i=1}^N \mathcal{D}_{KL}\big( \sum_{l=1}^{|\bm{\Theta}|} \frac{ \mathsf{P}_{\mathcal{M}(X_i, \theta_l) + \bm{B}(X_i, \theta_l)}}{|\bm{\Theta}|} \|  \sum_{j=1}^N \sum_{l=1}^{|\bm{\Theta}|} \frac{\mathsf{P}_{\mathcal{M}(X_j, \theta_l) + \bm{B}(X_j, \theta_l)}}{N|\bm{\Theta}|} )    \big).
 \label{discete_MI_1}
    \end{aligned}
\end{equation}
Now, we apply Lemma  \ref{lem:matrix_subsampling} w.r.t subsampling on $\theta$. For a subset $\bm{\Theta}_{\tau_3}$ of $\tau_3$ elements randomly selected from $\bm{\Theta}$, once $\tau_3 | |\bm{\Theta}|$, we have that (\ref{discete_MI_1}) is upper bounded by,
\begin{equation}
    \mathbb{E}_{\bm{\Theta}_{\tau_3}}  \frac{1}{N}\sum_{i=1}^N \mathcal{D}_{KL}\big( \sum_{\theta \in \bm{\Theta}_{\tau_3}} \frac{\mathsf{P}_{\mathcal{M}(X_i, \theta) + \bm{B}(X_i, \theta)}}{\tau_3} \|  \sum_{j=1}^N \sum_{\theta \in \bm{\Theta}_{\tau_3}} \frac{\mathsf{P}_{\mathcal{M}(X_j, \theta) + \bm{B}(X_j, \theta)}}{N\tau_3} ) \big).
\label{discrete_MI_2}
\end{equation}
With a similar idea, for a $\tau_2$-subset $\mathcal{X}^*_{\tau_2}$ randomly selected from $\mathcal{X}^*$, we have that (\ref{discrete_MI_2}) can be upper bounded by \begin{equation}
    \mathbb{E}_{\bm{\Theta}_{\tau_3},\mathcal{X}^*_{\tau_2}} \frac{1}{N}\sum_{i=1}^N \mathcal{D}_{KL}\big( \sum_{\theta \in \bm{\Theta}_{\tau_3}} \frac{ \mathsf{P}_{\mathcal{M}(X_i, \theta) + \bm{B}(X_i, \theta)}}{\tau_3} \| \sum_{\theta \in \bm{\Theta}_{\tau_3}, X_l \in \mathcal{X}^*_{\tau_2}}  \frac{\mathsf{P}_{\mathcal{M}(X_l, \theta) + \bm{B}(X_l, \theta)}}{\tau_2\tau_3} \big).
\label{discrete_MI_3}
\end{equation}
Finally, with another independent sampling on $\mathcal{X}^*$ to randomly produce a $\tau_1$-set $\mathcal{X}^*_{\tau_1}$, we obtain an unbiased estimation on (\ref{discrete_MI_3}) as 
\begin{equation}
    \mathbb{E}_{\bm{\Theta}_{\tau_3},\mathcal{X}^*_{\tau_2},\mathcal{X}^*_{\tau_1}} \frac{1}{\tau_1}\sum_{X_j \in \mathcal{X}^*_{\tau_1}} \mathcal{D}_{KL}\big( \sum_{\theta \in \bm{\Theta}_{\tau_3}} \frac{\mathsf{P}_{\mathcal{M}(X_j, \theta) + \bm{B}(X_j, \theta)}}{\tau_3} \|  \sum_{\theta \in \bm{\Theta}_{\tau_3}, X_l \in \mathcal{X}^*_{\tau_2}} \frac{\mathsf{P}_{\mathcal{M}(X_l, \theta) + \bm{B}(X_l, \theta)}}{\tau_2\tau_3} \big).
\label{discrete_MI_4}
\end{equation}
Now, if we replace each $\bm{B}(X,\theta)$ by  $\bm{B}(X,\theta) + \bm{B}_c$, for an independent Gaussian noise $\bm{B}_c \sim \mathcal{N}(0, c\cdot \bm{I}_d)$, we have that for arbitrary $X_j$ and mechanism $\mathcal{M}$, $\mathcal{D}_{KL}\big( \sum_{\theta \in \bm{\Theta}_{\tau_3}} \frac{\mathsf{P}_{\mathcal{M}(X_j, \theta) + \bm{B}(X_j, \theta)}}{\tau_3} \|  \sum_{\theta \in \bm{\Theta}_{\tau_3}, X_l \in \mathcal{X}^*_{\tau_2}} \frac{\mathsf{P}_{\mathcal{M}(X_l, \theta) + \bm{B}(X_l, \theta)}}{\tau_2\tau_3} \big)$ is uniformly upper bounded by $2{r}^2/{c}$. Therefore, we apply the Bernstein inequality to obtain the concentration control of such a mean estimation from  i.i.d. samples, where let $\mu$ to be (\ref{discrete_MI_4}) and 
$$ \Pr(\bar{\phi} - \mu \geq \beta) \leq exp\big({-\frac{0.5 m \beta^2}{ (2r^2/c)^2/\tau_1 + \beta/3 \cdot (2r^2/c) } }\big), $$
and therefore for a confidence $(1-\gamma)$, we may select $$ m \geq \frac{2\log(1/\gamma)}{\beta^2}\cdot \big((2r^2/c)^2/\tau_1 + \beta/3 \cdot (2r^2/c) \big).$$

\section{Proof of Theorem \ref{thm: composition}} \label{app: thm: composition}
We adopt the same notations as those used in the proof of Theorem \ref{thm: random_alg}. Without loss of generality, we consider two arbitrary randomized mechanisms $\mathcal{M}_1(X, \theta_1)$ and $\mathcal{M}_2(X, \theta_2)$ where $\theta_1$ and $\theta_2$ are independently selected. We construct a joint mechanism $\bar{\mathcal{M}}(X, \bar{\theta}) = (\mathcal{M}_1(X, \theta_1), \mathcal{M}_2(X, \theta_2))$, where $\bar{\theta}= (\theta_1, \theta_2)$, and a joint noise $\bar{B}=(B_1, B_2)$. Therefore, for any given $X_0$, the distribution of $\bar{\mathcal{M}}(X_0, \bar{\theta})$ is uniform in $\{\mathcal{M}_1(X_0, \theta_1), \theta_1 \in \bm{\Theta}_1\} \times \{\mathcal{M}_2(X_0, \theta_2), \theta_2 \in \bm{\Theta}_2\}$. 
With a similar reasoning as (\ref{on-average-KL}), we have that
\begin{equation}
    \begin{aligned}
      & \mathsf{MI}(X; \bar{\mathcal{M}}(X, \bar{\theta})+\bar{B}) = \mathbb{E}_{X_0}\mathcal{D}_{KL}(\mathsf{P}_{\bar{\mathcal{M}}(X_0,\bar{\theta})+\bar{B}} \| \mathsf{P}_{\bar{\mathcal{M}}(X,\bar{\theta})+\bar{B}} ) \\
      & \leq  \sum_{X_0, X'_0} \mathbb{P}_{X_0}\cdot \mathbb{P}_{X'_0} \cdot  \mathcal{D}_{KL}(\mathsf{P}_{\bar{\mathcal{M}}(X_0,\bar{\theta})+\bar{B}} \| \mathsf{P}_{\bar{\mathcal{M}}(X'_0,\bar{\theta})+\bar{B}} )                                     \\
      & \leq \sum_{X_0, X'_0} \sum_{{\theta}_{1j},\theta_{2l}} \frac{\mathbb{P}_{X_0}\cdot\mathbb{P}_{X'_0}}{|\bm{\Theta}_1| \cdot |\bm{\Theta}_2|} \cdot \mathcal{D}_{KL}(\mathsf{P}_{\bar{\mathcal{M}}(X_0,{\theta}_{1j},{\theta}_{2l}  )+\bar{B}} \| \mathsf{P}_{\bar{\mathcal{M}}(X'_0, \theta_{1{\pi}_{1} (X'_0,j)}, \theta_{2{\pi}_{2} (X'_0,l)}
 )+\bar{B}}) \\
      & \leq \sum_{X_0,X'_0}\sum_{{\theta}_{1j},\theta_{2l}}  \frac{\mathbb{P}_{X_0}\cdot  \mathbb{P}_{X'_0}}{|\bm{\Theta}_1| \cdot |\bm{\Theta}_2|}   \cdot \big( \frac{\|\mathcal{M}_1(X_0,\theta_{1j}) - \mathcal{M}_1(X'_0,\theta_{1\pi_1(X'_0,j)})  \|^2} {2\sigma_1^2} + \frac{\|\mathcal{M}_2(X_0,\theta_{2l}) - \mathcal{M}_2(X'_0,\theta_{2\pi_2(X'_0,l)})  \|^2} {2\sigma_2^2}  \big),
    \end{aligned}
\end{equation}
for arbitrary permutations $\pi_1(X'_0,\cdot)$ and $\pi_2(X'_0,\cdot)$ on $[1:|\bm{\Theta}_1|]$ and $[1:|\bm{\Theta}_2|]$. Here, we assume $B_1 \sim \mathcal{N}(0, \sigma^2_1 \cdot \bm{I})$ and $B_2 \sim \mathcal{N}(0, \sigma^2_2 \cdot \bm{I})$ and use the following fact of KL-divergence between two multivariate Gaussians 
\begin{align*}
 \mathcal{D}_{KL}\big( \mathsf{P}_{(\mu_1, \mu_2) + (B_1, B_2)} \| \mathsf{P}_{(\mu'_1, \mu'_2) + (B_1, B_2)} \big)  = \frac{1}{2} \big( \frac{\|\mu_1-\mu'_1\|^2}{\sigma^2_1} + \frac{\|\mu_2-\mu'_2\|^2}{\sigma^2_2}  \big).       
\end{align*}
Therefore, if we let $\psi_1(X_0, X'_0)$ equal the minimum-permutation distance $\mathsf{d}_{\pi_1}(\bm{y}_1(X_0), \bm{y}_1(X'_0))$, where $$\bm{y}_1(X_0) = (\mathcal{M}_1(X_0,\theta_1), \ldots , \mathcal{M}_1(X_0,\theta_{1|\bm{\Theta}_1|})), \bm{y}_1(X'_0) = (\mathcal{M}_1(X'_0,\theta_1), \ldots , \mathcal{M}_1(X'_0,\theta_{1|\bm{\Theta}_1|})),$$and we similarly define $\psi_2(X_0, X'_0)$ with $\pi_2$,  then we have a simplified upper bound on $ \mathsf{MI}(X; \bar{\mathcal{M}}(X, \bar{\theta})+\bar{B}) $ as follows,
\begin{equation}
   \mathsf{MI}(X; \bar{\mathcal{M}}(X, \bar{\theta})+\bar{B})  \leq \mathbb{E}_{X_0, X'_0}  \big(\frac{\psi_1(X_0,X'_0)}{2\sigma^2_1} + \frac{\psi_2(X_0,X'_0)}{2\sigma^2_2}\big). 
\end{equation}
Therefore, if we separately run Algorithm \ref{alg: random_MI} on $\mathcal{M}_1$ and $\mathcal{M}_2$ and determine $\sigma_1$ and $\sigma_2$ such that with confidence $(1-\gamma_1)$ and $(1-\gamma_2)$,
$$  \mathbb{E}_{X_0, X'_0}  \frac{\psi_1(X_0,X'_0)}{2\sigma^2_1} \leq v_1,  \mathbb{E}_{X_0, X'_0}  \frac{\psi_2(X_0,X'_0)}{2\sigma^2_2} \leq v_2,$$
respectively, then by union bound we have 
$\mathsf{MI}(X; \bar{\mathcal{M}}(X, \bar{\theta})+\bar{B}) \leq v_1+v_2$ with confidence at least $1-(\gamma_1+\gamma_2)$. The above analysis can be straightforwardly generalized to study any $T$-composition with independent randomness.

\section{Proof of Theorem \ref{thm: online_random_alg}} \label{app: thm: online_random_alg}

The proof is essentially a simple generalization of that of Theorem \ref{thm: random_alg} and Theorem \ref{thm: composition}. With a similar reasoning, for any $t \in [1:T]$, the mutual information $\mathsf{MI}(X;\mathcal{M}_1(X; \bar{\theta}_1)+\bm{B}_1, \ldots ,\mathcal{M}_t(X; \bar{\theta}_t)+\bm{B}_t)$ can be upper bounded by 
\begin{equation}
    \mathsf{MI}(X;\mathcal{M}_1(X; \bar{\theta}_1)+\bm{B}_1, \ldots ,\mathcal{M}_t(X; \bar{\theta}_t)+\bm{B}_t) \leq \mathbb{E}_{X_0, X'_0} \sum_{j=1}^t \frac{{\psi}_j(X_0, X'_0)}{2\sigma^2_j}.
\label{upper_bound_sequential}
\end{equation}
Here, ${\psi}_j(X_0, X'_0) = \|\bm{y}_j(X_0) - \bm{y}_j(X'_0)\|^2_2$, where $\bm{y}_j(X_0) = \big(\mathcal{M}_j(X_0,\bar{\theta}_{j1}), \mathcal{M}_j(X_0,\bar{\theta}_{j2}), \ldots \big)$, i.e., the concatenation of all possible outputs of $\mathcal{M}_j$ with different selections of randomness seeds $\bar{\theta}_j$, and $\bm{B}_j \sim \mathcal{N}(0, \sigma^2_j \cdot \bm{I})$. Therefore, the goal of Algorithm \ref{alg: sequential_random_MI} is essentially to determine $\sigma_t$ in the $t$-th round such that the increment of (\ref{upper_bound_sequential}), i.e., $\mathbb{E}_{X_0, X'_0} \frac{{\psi}_t(X_0, X'_0)}{2\sigma^2_t}$, is no larger than $v_t-v_{t-1}$ determined by the privacy budget prescribed.

The only difference is that in an online setup for adaptive composition, we need to ensure that the estimation in each step is accurate enough (without failure in any one of the steps with high probability). Therefore, for a $T$-composition, it suffices to ensure that the noise mechanisms determined in the $t$-th round can ensure current joint mechanism $\mathsf{MI}(X;\mathcal{M}_t(X; \bar{\theta}_t)+\bm{B}_t) \leq v_t$ with high confidence $(1-\gamma/T)$ and by the union bound, we can claim that $\mathsf{MI}(X;\mathcal{M}_t(X; \bar{\theta}_t)+\bm{B}_t) \leq v_t$ for any $t \in [1:T]$ with confidence at least $(1-\gamma)$. Applying Theorem \ref{thm: random_alg} for each step, we have the bound claimed.

\section{Proof of Theorem \ref{thm: lower-bound}}
\label{app: thm: lower-bound}
By Pinsker's inequality, we know the KL-divergence can be lower bounded by the total variation. Thus, in the following, we simply consider the  $\mathcal{D}_{TV}(\mathsf{P}_{X, \mathcal{M}(X)}\|\mathsf{P}_{X}\otimes \mathsf{P}_{\mathcal{M}(X)})$. We prove Theorem \ref{thm: lower-bound} by contradiction: if there exists some $y \in \mathcal{Y}$ such that $\mathbb{E}_{\bm{B} \sim Q(y)} \|\bm{B}\|_2 = o(r\sqrt{d}/\sqrt{v})$, then there exist some $\mathsf{P}_{X}$ and $\mathcal{M}$ such that $\mathcal{D}_{TV}(\mathsf{P}_{X, \mathcal{M}(X)}\|\mathsf{P}_{X}\otimes \mathsf{P}_{\mathcal{M}(X)})
>\sqrt{v}$ for $v =o(1)$. If not, we consider the following construction of $X$ and $\mathcal{M}(X)$. $X$ is randomly selected from a set $\{X_1, X_2\}$ and we consider a mechanism $\mathcal{M}$ which maps $X_1$ to $y$ and $X_2$ to $y'$. Thus, we have the following w.r.t. the total variation. 
\begin{equation}
    \begin{aligned}
    &\mathcal{D}_{TV}(\mathsf{P}_{X, \mathcal{M}(X)}\|\mathsf{P}_{X}\otimes \mathsf{P}_{\mathcal{M}(X)}) =  \mathbb{E}_{X_0 \sim \mathsf{D}}\mathcal{D}_{TV}(\mathsf{P}_{ \mathcal{M}(X_0)}\| \mathsf{P}_{\mathcal{M}(X)}) \\
    & = \frac{1}{4} \cdot \int_{z} \big( |\mathbb{P}_{y}(z) - \frac{1}{2}\cdot (\mathbb{P}_{y}(z)+\mathbb{P}_{y'}(z))| + |\mathbb{P}_{y'}(z) - \frac{1}{2}\cdot (\mathbb{P}_{y}(z)+\mathbb{P}_{y'}(z))| \big) dz \\
    & = \frac{1}{4} \cdot \int_{z} \big( | \mathbb{P}_{y}(z)-\mathbb{P}_{y'}(z)| + | \mathbb{P}_{y'}(z)-\mathbb{P}_{y}(z)| \big)dz\\
    & =  \mathcal{D}_{TV}\big(\mathcal{N}(y, \Sigma_{\bm{B}}(y))\| \mathcal{N}(y', \Sigma_{\bm{B}}(y')) \big).
    \end{aligned}
\end{equation}
Here, we simply use $\mathbb{P}_y(z)$ to denote the probability density of $Q(y)$ at point $z$. Now, we can apply the results about the lower bound of total variation for two multivariate Gaussians \cite{davies2022lower}, \cite{devroye2018total}, where
\begin{equation}
\label{tv_gaussian}
    \mathcal{D}_{TV}\big(\mathcal{N}(y, \Sigma_{\bm{B}}(y))\| \mathcal{N}(y', \Sigma_{\bm{B}}(y')) \big) = \Omega(1, \frac{\|y-y'\|^2}{\sqrt{(y-y')^T\Sigma_y(y-y)}}).
\end{equation}
Since $y'$ is free to select, we can set $y'-y$ to lie on the same direction of the $j$-th eigenvector of $\Sigma(y)$. On the other hand, as we assume $v=o(1)$, from the right side of  (\ref{tv_gaussian}),    it implies
$\frac{\|y-y'\|^2}{\sqrt{\lambda_j} \|y-y'\|} \leq \sqrt{v}.$
Since we can select $y'$ such that $\|y-y'\|=\Omega(r)$, therefore $\sqrt{\lambda_j} = \Omega(r/\sqrt{v})$ for any $j \in \{1,2,\ldots,d\}$. Therefore, $\mathbb{E}_{\bm{B}\sim Q(y)}[\|\bm{B}\|_2] = \Omega(r\sqrt{d}/\sqrt{v}),$ a contradiction.

\section{Proof of Theorem \ref{thm: generalzation_error}}
\label{app: thm: generalzation_error}
We prove the theorem by a contradiction. Without loss of generality we assume the distribution of $\mathcal{M}(X)$ is continuous and by definition,
\begin{equation}
\begin{aligned}
& \text{Gen}(\mathsf{D}, \mathcal{M}) = \mathbb{E}_{\bar{X} \perp X \sim \mathsf{D}, \mathcal{M}}\big(\mathcal{L}(\bar{X},\mathcal{M}(X))\big) - \mathbb{E}_{X \sim \mathsf{D}, \mathcal{M}} \big( \mathcal{L}(X, \mathcal{M}(X))\big) \\
& = \frac{1}{N} \sum_{i=1}^N \big( \frac{ \sum_{l=1}^N \int_{y} \mathcal{L}(X_i, y) \cdot \mathbb{P}_{\mathcal{M}(X_l)}(y) dy }{N} -  \int_y \mathcal{L}(X_i, y)\cdot \mathbb{P}_{\mathcal{M}(X_i)}(y) dy \big) > \sqrt{v/2}.
\end{aligned}
\label{average_1}
\end{equation}
With a rearrangement of the terms in (\ref{average_1}), we have that
\begin{equation}
    \int_{y} \big[ \sum_{i=1}^N \frac{\mathcal{L}(X_i, y)}{N} \cdot \big(     \frac{\sum_{l=1}^N \mathbb{P}_{\mathcal{M}(X_l)}(y)}{N} - \mathbb{P}_{\mathcal{M}(X_i)}(y) \big) \big] dy \geq \sqrt{v/2}.
\label{average_2}
\end{equation}
Since the output of loss function $\mathcal{L}$ is bounded and within $[0,1]$, (\ref{average_2}) suggests that 
\begin{equation}
    \frac{1}{N}\cdot \int_{y} \big[ \sum_{i=1}^N \big|     \frac{\sum_{l=1}^N \mathbb{P}_{\mathcal{M}(X_l)}(y)}{N} - \mathbb{P}_{\mathcal{M}(X_i)}(y) \big| dy \big] \geq \sqrt{v/2}.
\label{average_3}
\end{equation}

Now, we consider the inference task that an adversary sets out to identify the true selection of $X$ from $\mathcal{X}^*$ with observations $\mathcal{M}(X)$,  and we construct a strategy as follows. Given any observation $y=\mathcal{M}(X)$, the adversary will output $\tilde{X} = X_{i_0}$ such that $i_0 = \arg \max_{i} \mathbb{P}_{\mathcal{M}(X_i)}(y)$, with a {\em maximum a posteriori} estimation. 

It is clear that the optimal {\em a prior} success rate $(1-\delta^{\epsilon,\rho}_{o}) = 1/N$, while the posterior success $(1-\delta)$ of the above strategy is at least 
\begin{equation}
    \label{TV_advantage}
    \begin{aligned}
    & \frac{1}{N} \cdot \int_{y} \mathbb{P}_{\mathcal{M}(X_{i_0})}(y) dy = \frac{1}{N} \cdot \int_{y} \max_{i}\{ \mathbb{P}_{\mathcal{M}(X_{1})}(y),\ldots ,\mathbb{P}_{\mathcal{M}(X_{N})}(y) \} dy \\
    & = \frac{1}{N} \cdot \int_{y} \big( \frac{\sum_{l=1}^N \mathbb{P}_{\mathcal{M}(X_{l})}(y)}{N} +  \mathbb{P}_{\mathcal{M}(X_{i_0})}(y) - \frac{\sum_{l=1}^N \mathbb{P}_{\mathcal{M}(X_{l})}(y)}{N} \big)dy \\
    & \geq \frac{1}{N}+\frac{1}{N}\cdot \int_{y} \big[ \sum_{i=1}^N \big|     \frac{\sum_{l=1}^N \mathbb{P}_{\mathcal{M}(X_l)}(y)}{N} - \mathbb{P}_{\mathcal{M}(X_i)}(y) \big| dy \big] > \frac{1}{N} +\sqrt{v/2}.
    \end{aligned}
\end{equation}
In the last line of (\ref{TV_advantage}), we use the fact that given $\mathbb{P}_{\mathcal{M}(X_{i_0})}(y)$ is the maximal across $\mathbb{P}_{\mathcal{M}(X_{i})}(y)$, $\mathbb{P}_{\mathcal{M}(X_{i_0})}(y)-\frac{\sum_{l=1}^N \mathbb{P}_{\mathcal{M}(X_{l})}(y)}{N}$ is also the maximal across all the possible $|\mathbb{P}_{\mathcal{M}(X_{i})}(y)-\frac{\sum_{l=1}^N \mathbb{P}_{\mathcal{M}(X_{l})}(y)}{N}|$ for any $i=1,2,\cdots, N$. Thus, based on the maximal is no less than the average of a set of numbers, we have the lower bound in  (\ref{TV_advantage}). 

Now, by Pinsker's inequality, if $\Delta_{KL}\delta \leq v$, then $\Delta_{TV}\delta \leq \sqrt{v/2}.$ Therefore, $\delta^{\epsilon,\rho}_o - \delta \leq \sqrt{v/2},$ which leads to a contradiction.

\section{Local PAC Privacy}
\label{sec:local}
In this section, we generalize PAC Privacy to a local/decentralized setup. The first application is similar to local DP \cite{shuffling}, \cite{LDP2011}, where a user does not trust any other parties and plans to publish her perturbed sample. The perturbed released individual data could then be collected and postprocessed, which does not incur additional privacy loss. In a scenario where a user randomly samples $x$ from her local sample set $\mathsf{U}$ and publishes, we may select $\mathcal{M}$ to be the identity function, i.e., $\mathcal{M}(x)=x$ and measure the corresponding PAC Privacy. 

\begin{Ex}[Data Publishing via Local DP and PAC Privacy] We suppose $\mathsf{U}$ to be CIFAR10 $\cite{cifar10}$, a canonical test dataset commonly used in computer vision, which consists of $60,000$ $32\times32$ color images. We normalize each pixel to within $[0,1]$ and take each sample as a 3072-dimensional vector. Let the local data publishing be $\mathcal{M}(x)=x$ where $x$ is randomly sampled from CIFAR10. With a target bound $\mathsf{MI}(x; \mathcal{M}(x)+\bm{B}) \leq 1$, given $l_{\infty}$-norm sensitivity, the adversarial worst-case local DP requires $\bm{\Theta}(d)$-large noise, where in this case $\mathbb{E}[\| \bm{B}\|_2] = d/\sqrt{2} = 2.2\times 10^3$ via the Gaussian Mechanism $\cite{bun2016concentrated}$. On the other hand, via the upper bound in Theorem \ref{thm:deter_m}, for PAC Privacy, we only need noise $\mathbb{E}[\|\bm{B}\|_2] \leq 1.6 \times 10^2$, $13\times$ smaller than that of local DP. 
\label{ex: local}
\end{Ex}

The case where we consider a decentralized system of $n$ semi-honest users is more complicated. Each of the users will provide some local data $x_i$. With $X=(x_1, x_2, \ldots , x_n)$, they set out to collaboratively implement some mechanism $\mathcal{M}(x_1, \ldots ,x_n),$ and the output $\mathcal{M}(x_1, \ldots ,x_n)$ is the only {\em necessary} information they need to disclose. Meanwhile, they wish to minimize the leakage of $x_i$s from the exposure of $\mathcal{M}(x_1, \ldots ,x_n)$.

Recall the PAC Privacy framework, which is basically three parts: (a) through simulation, we determine the parameters of noise distribution $Q(X, \theta)$; (b) we generate $X$ and $\theta$, and evaluate $\mathcal{M}(X, \theta)$; (c) we generate $\bm{B}(X, \theta) \sim Q(X, \theta)$ and publish perturbed $\mathcal{M}(X, \theta)+\bm{B}$. However, without a trusted server, a straightforward implementation of 
this protocol is problematic. First, the mechanism $\mathcal{M}'$ we need to analyze now integrates the three procedures (a)-(c) mentioned above, where all intermediate simulations may cause additional privacy loss. Second and more important, to analyze the privacy of the new mechanism $\mathcal{M}'$, we then need $\mathcal{M}''$ to incorporate both $\mathcal{M}'$ and its new privacy analysis simulation, which turns into an endless loop.  

To this end, we need a secure computation for all three procedures (a)-(c) to fill the gap between local and centralized PAC Privacy. Fortunately, if we slightly relax the threat model, Secure Multi-party Computation (MPC) \cite{cramer2015secure} provides a generic solution. Theoretically, MPC enables a group of users to compute any generic function, and meanwhile it ensures that besides the output of the function and 
a (possibly colluding) user's own input, there is negligible additional information leakage to any polynomial-computation-bounded adversary\footnote{Though MPC can ensure negligible information leakage in decentralized computing, there is no guarantee for the disclosure of output $\mathcal{M}(X)$. Therefore, we need a notion such as PAC Privacy for generic privacy measurement.}. Therefore, when the adversary is relaxed to be computationally-bounded, the group of semi-honest users \footnote{We need to assume users are semi-honest to ensure that the data generation follows the prescribed distribution $\mathsf{D}$, or additional verifiable computation on data generation.} may implement MPC for both privacy analysis simulation and the computation of $\mathcal{M}(X)(+\bm{B})$, where the output finally disclosed is a single (noisy) evaluation $\mathcal{M}(X)(+\bm{B})$ (if noise is necessary). We formally describe such {\em Locally-collaborative PAC Individual Privacy} as follows:
\begin{Def}[$(\delta, \rho, n, \mathsf{D})$ Locally-collaborative PAC Individual Privacy]
For a data processing mechanism $\mathcal{M}$ and $n$ semi-honest computationally-bounded users, who $($collaboratively$)$ generate a dataset $X=(x_1, x_2, \ldots ,x_n)$ from some data distribution $\mathsf{D}$ where $x_i$ is from the $i$-th user, we say $\mathcal{M}$ satisfies $(\delta, \rho, n, \mathsf{D})$-locally-collaborative PAC Individual Privacy if for any $i \not = j \in [1:n]$ the following experiment is impossible:

A dataset $X = (x_1, x_2, \ldots ,x_n)$ is generated from distribution $\mathsf{D}$ and $\mathcal{M}(X)$ is published. The $j$-th user who knows $x_j$, $\mathsf{D}$ and $\mathcal{M}$, $j \not=i$, is asked to return an estimation $\tilde{x}_i$ on $x_i$ such that with probability at least $(1-\delta)$, $\rho(\tilde{x}_i, x_i) =1$. 
\label{def:pac_local_individual}
\end{Def}

Compared to the centralized case, the only difference in the analysis is that we need to control the conditional $\mathsf{MI}(x_i;\mathcal{M}(X)|x_j)$ for any $i\not= j \in [1:n]$ rather than simply $\mathsf{MI}(x_i;\mathcal{M}(X))$. We need to ensure that the $j$-th user who already knows her own data $x_j$ still cannot infer too much about another $i$-th user's data from $\mathcal{M}(X)$. By the chain rule, we have that
$$\mathsf{MI}(x_i;\mathcal{M}(X)|x_j) \leq \mathsf{MI}(X;\mathcal{M}(X)) - \mathsf{MI}(x_j; \mathcal{M}(X)) \leq \mathsf{MI}(X;\mathcal{M}(X)),$$
for any $i$ and $j$. Therefore, $\mathsf{MI}(X; \mathcal{M}(X))$ is still an upper bound on the posterior advantage in the  locally-collaborative version. As a final remark, the above semi-honest assumption can be further relaxed to the adversarial setup with additional applications of  unconditionally secure MPC \cite{unconditional} and verifiable randomness \cite{verifiable}. The need of verifiable randomness is to ensure the distributions of data generation and randomness seed selection are consistent in both privacy analysis steps (a,b) and output generation steps (c) for each user.  

\begin{Rmk}
Though $\mathsf{MI}(X; \mathcal{M}(X))$ is an upper bound on both $\max_{i}\mathsf{MI}(x_i;\mathcal{M}(X))$ and  $\max_{i,j}\mathsf{MI}(x_i;\mathcal{M}(X)|x_j)$, which can produce security parameters of centralized and local individual PAC Privacy in Definitions \ref{def:pac_individual} and \ref{def:pac_local_individual}, respectively, we must stress that, using $\mathsf{MI}(X; \mathcal{M}(X))$, we also obtain stronger and more meaningful PAC Privacy for {\em the whole set $X$}. Take membership inference as an example. If each user decides to randomly participate in the analysis with her local data $x_i$, $\mathsf{MI}(X; \mathcal{M}(X))$ quantifies the maximal posterior advantage for any possible inference the adversary could implement on all participants rather than that of a single one.     
\end{Rmk}

\end{document}